\documentclass[12pt,a4paper,notitlepage]{article}
\usepackage[utf8]{inputenc}
\usepackage[english]{babel}
\usepackage[T1]{fontenc}
\usepackage{hyperref}
\usepackage{amsmath}
\usepackage{stmaryrd}
\usepackage{amsfonts}
\usepackage{amssymb}
\usepackage{color} 
\usepackage{caption}
\usepackage{amsmath}
\usepackage{relsize,exscale}
\usepackage{array,multirow,makecell}

\usepackage[toc,page]{appendix}
\usepackage{amsthm}

\setcellgapes{1pt}
\makegapedcells
\newcolumntype{R}[1]{>{\raggedleft\arraybackslash }b{#1}}
\newcolumntype{L}[1]{>{\raggedright\arraybackslash }b{#1}}
\newcolumntype{C}[1]{>{\centering\arraybackslash }b{#1}}

\newcommand{\mI}{\mathcal{I}}
\newtheorem{theorem}{Theorem}
\newtheorem{definition}{Definition}
\newtheorem{property}{Property}
\newtheorem{proposition}{Proposition}

\newtheorem{corollary}{Corollary}

\newcommand{\sym}{\mathrm{Sym}}

\newtheorem{lemma}{Lemma}
\usepackage{enumerate}
\usepackage{subfig}
\usepackage{graphicx}

\newcommand{\SU}{\mathrm{SU}}
\newcommand{\SL}{\mathrm{SL}}
\newcommand{\SO}{\mathrm{SO}}
\newcommand{\J}{\mathrm{J}}

\def\extd{\mathrm {d}}

\usepackage{graphicx}

\newcommand{\beq}{\begin{equation}}
\newcommand{\eeq}{\end{equation}}
\newcommand{\bea}{\begin{eqnarray}}
\newcommand{\eea}{\end{eqnarray}}
\definecolor{mygray}{gray}{0.3}

\newcommand{\bes}{\begin{eqnarray}}
\newcommand{\ees}{\end{eqnarray}}

\newcommand\restr[2]{{
  \left.\kern-\nulldelimiterspace 
  #1 
  \vphantom{\big|} 
  \right|_{#2} 
  }}
\def\extd{\mathrm {d}}
\newcommand{\U}{\mathrm{U}}
\usepackage{multicol}
\usepackage{amssymb}

\usepackage[left=2.4cm,right=2.4cm,top=2.5cm,bottom=2.5cm]{geometry}

\begin{document}


\begin{center}
\textbf{\large{No Ward-Takahashi identity violation for Abelian tensorial group field theories with a closure constraint}}\\

\medskip
\vspace{15pt}

{\large Vincent Lahoche$^a$\footnote{vincent.lahoche@cea.fr}  \,\, B\^em-Bi\'eri Barth\'el\'emy Natta$^b$\footnote{nattabarth@gmail.com}\,\,
Dine Ousmane Samary$^{a,b}$\footnote{dine.ousmanesamary@cipma.uac.bj}}
\vspace{15pt}

a)\,  Commissariat à l'\'Energie Atomique (CEA, LIST),
 8 Avenue de la Vauve, 91120 Palaiseau, France

b)\, Facult\'e des Sciences et Techniques (ICMPA-UNESCO Chair),
Universit\'e d'Abomey-
Calavi, 072 BP 50, B\'enin
\vspace{0.5cm}
\begin{abstract}
\noindent 

This paper aims at investigating the nonperturbative functional renormalization group for tensorial group field theories
with nontrivial kinetic action and closure constraint. We consider the quartic melonic just-renormalizable $[U(1)]^6$ model and show that due to this closure constraint the first order Ward-Takahashi identity takes the  trivial form as for the models with propagators proportional to identity. We then construct the new version of the effective vertex expansion applicable to this class of models, which in particular allows to close the hierarchical structure of the flow equations in the melonic sector. As a consequence, there are no additional constraints on the flow equations, and then we can focus on the existence of physical Wilson-Fisher fixed-points solutions in the symmetric phase.

\medskip


\noindent
\textbf{pacs:} {71.70.Ej, 02.40.Gh, 03.65.-w}
\end{abstract}
\end{center}



\section{Introduction}
The quantum theory of gravitation or simply quantum gravity (QG), aiming to provide a unification between general relativity and quantum physics, is probably one of the major challenges of the century. Such a program is motivated as much by pure intellectual satisfaction as by the existence of real contact points, where quantum effects can no longer be neglected in the description of gravitational effects. This is notably the case where gravitational effects are strong enough, in the vicinity of black holes or similar compact astrophysical objects where the effects of gravity are strong, such as neutron stars, or below the Planck scale in the early universe ( see \cite{Jegerlehner:2021vqz}-\cite{Donoghue:1994dn} and references therein). Note that there isn't just one path to quantum gravity right now and a host of approaches compete with each other, or mutually enrich each other. Currently, the most active approach in terms of the number of researchers is string theory (the list of the references is so long and we consider just \cite{Witten:2002wb}). However, other approaches have emerged from the ideas of quantification of gravity, the most dominant being the loop quantum gravity (LQG) \cite{Rovelli:1997yv}-\cite{Rovelli:1998gg}. Group field theory (GFT) is one of the most popular syntheses of the canonical approach proposed by LQG. Historically it has been proposed as a field theoretical framework able to reproduce spinfoams amplitudes; without ambiguity on their relative weights, and thus suitable to discuss arbitrary large sums of such a spinfoam aiming to define a continuum limit \cite{Ambjorn:1992eh}-\cite{Oriti:2009nd}. It has also been shown that the GFT formalism corresponds to a second unification of the standard LQG \cite{Rovelli:2011eq}-\cite{Markopoulou:2006qh}. In general, it was essential that having a formalism of field theory, and all the associated mathematical ingredients (for instance the renormalization) was great progress concerning the previous approaches, in particular in the perspective of describing a continuous limit. The underlying hope is to find a fundamental quantum theory that well describes the microscopic states of space-time; which are assumed to behave collectively following the Einstein equation of gravitation to some limit. Several steps have been taken in this direction, notably in the description of the primordial universe, and for the calculation of the entropy of an isolated horizon, \cite{Steinhauer:2015saa}.
\medskip

In the absence of experiments, the best we may expect is a class of theories, all of which converge toward general relativity on a large scale. The most important tool to discuss universality about some scaling is probably the renormalization group (RG). The RG program for the quantum gravity model can be useful to explore the physical content of the theories, beyond the limits of exact methods. Such a program however has the advantage to be helpful: If (i) RG equations are easy to compute. (ii) If the approximations used to solve the flow equations are reliable. Therefore, before investigating the RG, the first ingredient we have chosen to do is to develop robust and reliable methods, allowing us to explore a large domain of the specific phase spaces of the theoretical models. This paper, like the ones cited on the same subject, aims to investigate such a methodology.
\medskip

Tensorial group field theories (GFTs) are a recent proposal to tackle the issue of universality. They provide an example of mutual enrichment, as they combine ideas arising from random tensors and GFTs. In particular, they assume the idea that a special kind of invariance is required to make the theory power-countable. The existence of such a power counting in turn is the basic ingredient to define a renormalization group flow, suitable to describe large scale physics and phase transitions \cite{Rivasseau:2014ima}-\cite{Lahoche:2019cxt}. This special kind of invariance is called \textit{tensoriality}  in the literature \cite{Gurau:2011xp}-\cite{Delporte:2019tof}, and refers to the invariance of interactions for some continuous internal symmetry group. Moreover, it has been proved that endow GFTs from such an invariance, discard singular contributions in the corresponding spinfoams amplitudes \cite{Ooguri:1991ib}-\cite{BenGeloun:2011jnm}. The existence of phase transitions is at the heart of the commonly accepted scenario called ‘geometrogenesis', corresponding to the condensation of an arbitrarily large number of microscopic degrees of freedom, as for quantum fluids, expected to behave as a (semi-) classical space-time \cite{Carrozza:2014rba}-\cite{Lahoche:2019cxt}. Thus, investigations using renormalization group (RG)  techniques, besides the mathematical requirement to have a well-defined theory at the microscopic scale (Bogoljubow-Parasjuk-Hepp-Zimmermann BPHZ theorems), aims to describe such condensation phenomena or to say if such a transition exists and on what parameters of the model are existence depends.

\medskip

In the last decade, renormalization for TGFTs has been largely investigated; and has follows two complementary trajectories. Purely mathematical investigations through BPHZ theorems and constructive expansions \cite{Carrozza:2013wda}-\cite{Lahoche:2015ola} aiming to endow the framework with a solid mathematical background, and numerical investigations using the Wetterich non-perturbative RG formalism \cite{Carrozza:2014rba}-\cite{Lahoche:2019cxt}. Note that the last one does not only concern TGFT, but other discrete gravity approaches like matrix and tensor models have been investigated using non-perturbative techniques \cite{Brezin:1992yc}-\cite{Eichhorn:2017xhy} and references therein. The existence of Wilson-Fisher fixed points reminiscent of a second-order phase transition was a universal feature of all the most likely investigated models. Such a 'super' universality seemed to be a basic feature of the specific non-locality used to define TGFT's interactions, but on the other hand, did not provide any additional information likely to reduce the number of possible candidates \cite{Geloun:2016qyb}-\cite{Carrozza:2017vkz}; beyond, for example, properties like asymptotic freedom, which also seemed an almost universal characteristic \cite{Geloun:2016qyb}-\cite{Benedetti:2014qsa}. However, in a series of papers \cite{Lahoche:2018ggd}-\cite{Lahoche:2019cxt} see also
\cite{Lahoche:2019ocf}-\cite{Lahoche:2020pjo}, it has been shown that the current methodology used to solve the renormalization group equations could be improved in two ways. First, summing the most relevant sectors from the melonic one following a novel strategy called effective vertex expansion (EVE) \cite{Lahoche:2018oeo}. Second, taking into account Ward-Takahashi (WT) identities in the construction of the approximate solution of the RG equation. The EVE provides a non-trivial way to close the infinite hierarchy of RG equations, and it has been proved that it extends maximally the landscape in the vicinity of the Gaussian fixed point in contrast with the crude truncation method. Taking into account the constraint coming from the WT identities it has been proved that all the non-Gaussian fixed points violate this constraint and no fixed points are expected at all (particularly for the quartic model \cite{Lahoche:2018ggd}). From now the existence of a physical fixed point must require taking into account the natural constraints of the model given with the WT identities. In our previous investigation, we considered only models without Gauss closure constraint and we have already suspected that this closure constraint could make the Ward constraint obsolete, that is to say without effect on the Wetterich flow equations.
In this manuscript we intend to respond to this assertion i.e. we investigate the compatibility with WT identities and the FRG solutions for $U(1)$-models and provide the argument in favor of this Gauss constraint for emergent quantum gravity models. Let us recall two essential points. First, the closure constraint is an additional symmetry constraint that is necessary, for instance, to interpret the Feynman amplitudes of the tensor model as the amplitudes of a discretized simplicial manifold issued from topological BF theories \cite{Dijkgraaf:1989pz}-\cite{Putrov:2016qdo}.  Also will help for the emergence of a well-defined metric on the space after possible phase transition. Second, once again, due to the closure constraint the usual version of the EVE as proposed in \cite{Lahoche:2018ggd} need to be improved (the notion of the skeleton will be replaced by the topological skeleton) and we focus on the phase symmetry expansion.

\medskip

The paper is organized as follows: In section \eqref{sec2} we provide a useful ingredient concerning the TGFTs formalism, the definitions which will be important throughout the manuscript. We also give a brief introduction on the FRG using the Wetterich formalism. In section \eqref{secEVE} we study the FRG for our model. First, we establish the WT identities with new derivative techniques due to the appearance of the delta function in the propagator. We show that these identities are equivalent to the symmetric phase condition. The new version of the EVE is also given allows closing the hierarchical structure of the flow equations in the melonic sector. We also investigate the numerical solution of these flow equations and compare these with the one obtained with the usual truncation method. Section \eqref{sec4} is devoted to the conclusion and remarks. We give some appendices to elucidate the detailed computations using throughout the manuscript.

\section{Technical preliminaries}\label{sec2}

The GFTs are statistical field theories on the group manifolds rather than standard quantum field theories on euclidean space. For the quantum gravity, the standard choices are $\SU(2)$, $\SL(2, \mathbb{R})$, $\SL(2, \mathbb{C})$, $\SO(4)$ or $\SO(3,1)$, depending if we work in $3$ or $4$-dimensional gravity, in euclidean or minkowskian signature. We expect however that the universal features about RG flow depends essentially on dimensions and compactness of the group manifolds, and weakly on the Lie algebra structure. In the UV regime, where our investigations work, one expects moreover that compactness is an irrelevant parameter, and the essential ingredient is only the dimension of the Lie group, i.e. the number of Lie generators. Therefore, the behavior of the RG flow will be essentially the same for $\mathrm{G}=\SU(2)$ and $\mathrm{G}=(\U(1))^3$, replacing the non-Abelian group $\SU(2)$ by $3$ copies of the compact Abelian group $(U(1))^3$. For instance, the characteristics of the power counting are expected to be the same, which can be established rigorously, see \cite{Carrozza:2013wda}-\cite{Lahoche:2015ola}. Then, the renormalization properties can be proved rigorously, and in this paper, we will essentially focus on a just-renormalizable abelian GFT based on the group $\U(1)$. However, some investigations will also be carried out for non-abelian versions. We provide a short presentation of the (tensorial) GFT formalism in this section. We especially recall all the elementary definitions and statements useful for the rest of the paper, and we conclude with a presentation of the (FRG) formalism.

\subsection{Tensorial group field theory formalism}\label{sec21}
We consider a complex field $\varphi$, $\bar{\varphi}$ over $d$-copies of a compact Lie group $\mathrm{G}$,
\begin{equation}
\varphi, \bar{\varphi} : \mathrm{G}^{\times d}\rightarrow \mathbb{C}\,.
\end{equation}
The configurations of the complex field $\varphi$ are assumed to be randomly distributed, and the generating functional:
\begin{equation}
Z[\J, \bar{\J}]:= \int [\extd\varphi] [\extd\bar{\varphi}] \exp \left(-S[\varphi,\bar{\varphi}]+\int \extd\textbf{g} \bar{\J}(\textbf{g}) \varphi(\textbf{g})+ \int \extd\textbf{g} \bar{\varphi}(\textbf{g}) \J(\textbf{g})\right)\,,\label{partition1}
\end{equation}
where $\J, \bar{\J}: \mathrm{G}^d\to \mathbb{C}$ are complex source fields; $\textbf{g}:=(g_1,\cdots, g_d)\in \mathrm{G}^d$, and $\extd\textbf{g}:= \prod_{i=1}^d \, \extd g_i$, '$\extd g$' being the standard normalized Haar measure on the group manifold $\mathrm{G}$ ($\int dg=1$). In the quantum gravity context, the field $\varphi$ is endowed with a discrete $(d-1)$-simplex interpretation; and perturbative expansions are expect to reproduce spin-foam amplitudes, labeled with $2$-complex graphs dual to $d$-dimensional cellular decomposition. Because such a cellular decomposition is made of a set of glued $d$-simplices, action $S$ must have to encode:
\begin{itemize}
\item How to glue $(d-1)$-simplices to build $d$ simplices,
\item How to glue $d$-simplices.
\end{itemize}
These two requirements ensure that interactions in $S$ cannot be local in the usual sense. However, it has been showed that among the different prescriptions that we can imagine, there is one that looks better than the others called \textit{tensorial invariance}, especially in regard to properly construct a RG program. To briefly review it, note that there is a natural symmetry acting on the components $\varphi(\textbf{g})$ of the complex field. At the classical level, for $\varphi \in L^2(\mathrm{G}^d)$, it corresponds to unitary transformations $\mathcal{U}: L^2(\mathrm{G}^{\otimes d}) \to L^2(\mathrm{G}^{\otimes d})$, acting \textit{independently} on each group argument of fields:
\begin{equation}
\mathcal{U}[\varphi](\textbf{g}):= \int \extd\textbf{g}^\prime \left[\prod_{i=1}^d U^{(i)}(g_i,g_i^\prime) \right] \varphi(\textbf{g}^\prime)\,.
\end{equation}
A tensorial interaction $\mathcal{V}[\varphi,\bar{\varphi}]$ is invariant under such a transformation,
\begin{equation}\label{eq3}
\mathcal{V}[\mathcal{U}[\varphi],\mathcal{U}[\bar{\varphi}]]=\mathcal{V}[\varphi,\bar{\varphi}]\,.
\end{equation}

Moreover, any suitable interaction is assumed to admits an expansion in power of fields, so that $\mathcal{V}[\varphi,\bar{\varphi}]$ can be viewed as a sum of monomials. Tensorial invariance requires that each terms of the expansion must have the same number of fields $\varphi$ and $\bar{\varphi}$, say for instance $n$. To be unitary invariant, interaction has to be such that group arguments are contracted pairwise between fields $\varphi$ and $\bar{\varphi}$ with Dirac delta $\delta(g_i(\bar{g}_i)^{-1})$, the notation $\bar{g}_i$ referring to group arguments of the conjugate field $\bar{\varphi}$. As an example:
\begin{equation}
\mathcal{V}^{(4)}_{\text{melo}}[\varphi,\bar{\varphi}]= \int \extd\textbf{g}\extd\textbf{g}^\prime\, \varphi(g_1,g_2,g_3) \bar{\varphi}(g_1,g_2^\prime,g_3^\prime) \varphi(g_1^\prime, g_2^\prime, g_3^\prime) \bar{\varphi}(g_1^\prime, g_2, g_3)\,. \label{melonexample}
\end{equation}
Tensorial interactions can be conveniently pictured as $d$-colored bipartite regular graphs as follow:
\begin{enumerate}
\item We associate black and white nodes to the fields $\varphi$ and $\bar{\varphi}$ respectively,
\item We hook $d$ colored open edges to each node, materializing the $d$ arguments $g_1,\cdots, g_d$,
\item We hook the colored edges pairwise between black and white nodes accordingly to their colors.
\end{enumerate}
Figure \ref{fig1} provides some examples for $d=3$. This tensorial invariance requirement defines our choice of the \textit{theory space}. Note that a tensorial invariant can be connected or not following what its corresponding graph is. The second graph on Figure \ref{fig1} provides an example of a disconnected graph, and we have the following definition:
\begin{definition}
We call bubbles the connected tensorial invariants, whose corresponding graphs are made of a single connected component.
\end{definition}
This definition is especially important in regard to renormalization, because it allows to define a locality prescription \cite{Carrozza:2016tih}-\cite{Lahoche:2019cxt}:
\begin{definition}\label{locality}
For GFTs, connected bubbles are said to be locals. In the same way any functional $\mathcal{V}[\varphi,\bar{\varphi}]$ which expands in terms of bubbles is said to be a local functional.
\end{definition}
\begin{figure}
\begin{center}
\includegraphics[scale=1.2]{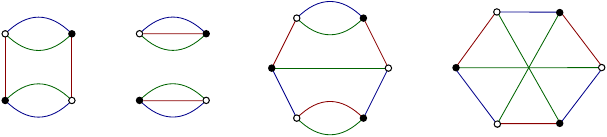}
\end{center}
\caption{Examples of tensorial interactions in dimension $d=3$. The first one on left corresponds to the formula \eqref{melonexample}.}\label{fig1}
\end{figure}
If the action $S$ expands only in terms of tensorial invariant, the corresponding GFT is simply tensorial. For renormalization investigations, however, it is suitable to break the global unitary invariance at the kinetic level, and we choose for the kinetic action:
\begin{equation}
S_{\text{kin}}[\phi, \bar{\phi}]=\int_{\mathrm{G}^{d}} \extd\textbf{\textbf{g}}\, \bar{\varphi}({\textbf{g}} )\left(-\Delta_\textbf{g}+m^2\right)\varphi({\textbf{g}} )\,, \label{Skin}
\end{equation}
where $\Delta_\textbf{g}$ is the Laplace-Beltrami operator over the group manifold $\mathrm{G}^{d}$. Except the kinetic contribution, all the remaining interactions (involving more than two fields) will taken to be unitary invariant. In the next section we provide some basics about Feynman graphs, power counting and renormalization. As we will recall, the most divergent sector say \textit{melonic} is built from a special class of Feynman graphs called \textit{melons}. These melons in turns have to made only with \textit{melonic bubbles}. We have the following definition:
\begin{definition}\label{defmelons}
Any melonic bubble $b_\ell$ of valence $\ell$ may be deduced from the elementary melon $b_1$:
\begin{equation}
b_1:=\vcenter{\hbox{\includegraphics[scale=1]{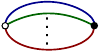} }}\,,
\end{equation}
replacing successively $\ell-1$ colored edges as follows, defining the $(d-1)$-dipole insertion operator $\mathfrak{R}_{i}$:
\begin{equation}
\vcenter{\hbox{\includegraphics[scale=1]{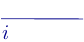} }}\underset{\mathfrak{R}_{i}}{\longrightarrow}\vcenter{\hbox{\includegraphics[scale=1]{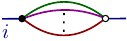} }}\,.
\end{equation}
In formula: $b_\ell := \left(\prod_{\alpha=1}^{\ell-1}\mathfrak{R}_{i_\alpha}\right) b_1$.
\end{definition}
The first and third bubbles on Figure \ref{fig1} are melons. For our nonperturbative investigations, we especially focus on a sub-sector of the melons, said \textit{non-branching}:
\begin{definition}\label{defnonbranch}
A non branching melonic bubble of valence $\ell$, $b_\ell^{(i)}$ is labeled with a single index $i\in\llbracket 1,6\rrbracket$, and defined such that:
\begin{equation}
b_\ell^{(i)}:= \left(\mathfrak{R}_{i}\right)^{\ell-1}\,b_1\,.
\end{equation}
\end{definition}
Figure \ref{fig2} provides the generic structure of melonic non-branching bubbles.

\begin{center}
\begin{equation*}
\vcenter{\hbox{\includegraphics[scale=1]{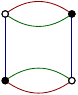} }} \,\underset{\mathfrak{R}_{i}}{\longrightarrow}\, \vcenter{\hbox{\includegraphics[scale=1]{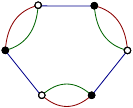} }}\,\cdots \underset{\mathfrak{R}_{i}}{\longrightarrow}\, \vcenter{\hbox{\includegraphics[scale=1]{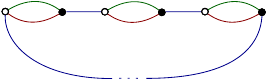} }}\underset{\mathfrak{R}_{i}}{\longrightarrow}\cdots
\end{equation*}
\captionof{figure}{Structure of the non-branching melons, from the smallest one $b_2$.} \label{fig2}
\end{center}
Finally, we focus in this paper on GFT endowed with a specific gauge invariance called \textit{closure constraint} or Gauss constraint. As we mention in the introduction, this constraint comes from LQG, and is expressed as
\begin{equation}
\varphi(g_1,\cdots, g_d) = \varphi(g_1h,\cdots, g_dh) \,,\quad \forall h\in \mathrm{G}^d\,.
\end{equation}
This condition has to be included in the definition of the partition function \eqref{partition1}, but we may carefully introduce this because the propagator corresponding to the kinetic action \eqref{Skin} is not invertible from projection into the gauge-invariant fields subspace. A common way to solve this difficulty is to define the (gauge-invariant) bare propagator $C(\textbf{g},\textbf{g}^\prime)$ as:
\begin{equation}
C(\textbf{g},\textbf{g}^\prime)= \left(\hat{P} \frac{1}{-\Delta_\textbf{g}+m^2} \hat{P}\right)(\textbf{g},\textbf{g}^\prime)\,,\label{propaprojection}
\end{equation}
where the projector $P$ is defined as:
\begin{equation}
\hat{P}(\textbf{g},\textbf{g}^\prime):= \int \extd h \prod_{i=1}^d \delta(g_i^{\prime} h g_i^{-1})\,,
\end{equation}
which becomes important in the definition of the Feynman amplitude from Wick theorem. It is formally equivalent to write the partition function \eqref{partition1} as:
\begin{equation}
Z[\J, \bar{\J}]= e^{\int \extd \textbf{g} \extd\textbf{g}^\prime \frac{\delta }{\delta \varphi(\textbf{g})} C(\textbf{g},\textbf{g}^\prime) \frac{\delta }{\delta \bar{\varphi}(\textbf{g}^\prime)}}\, \exp \left(-S_{\text{int}}[\varphi,\bar{\varphi}]+\int \extd\textbf{g} \bar{\J}(\textbf{g}) \varphi(\textbf{g})+ \int \extd\textbf{g} \bar{\varphi}(\textbf{g}) \J(\textbf{g})\right)\bigg\vert_{\varphi=\bar{\varphi}=0}\,,\label{partition2}
\end{equation}
where $S_{\text{int}}$ contain only monomials of valence higher than $2$. This relation may be simply obtained and for more detail see \cite{Lahoche:2015ola}. Now we provide the notion of Feynman graphs and the power counting of our model.

\subsection{Feynman diagrams and power counting}

The partition function \eqref{partition2} as well as correlations functions expand perturbatively in power series involving Feynman amplitudes $\mathcal{A}_{\mathcal{G}}$ generally are divergent. The Feynman graphs $\mathcal{G}$ indexing the amplitudes of the perturbation theory are $2$-simplexes, i.e. a set of vertices, edges and faces: $\mathcal{G}:=(\mathcal{V}, \mathcal{L}\cup\mathcal{L}_{\text{ext}},\mathcal{F}\cup\mathcal{F}_{\text{ext}})$. We denote as $V$, $L$, $L_{\text{ext}}$, $F$ and $F_{\text{ext}}$ respectively the number of vertices, internal edges, external edges, internal faces and external faces. Note that we attribute the color zero for the propagator edges. Moreover, the edges and faces are of two types, internals and externals, the last ones being labeled with an index ‘‘$\text{ext}$''. We recall that the faces are defined as follow:
\begin{definition}
A face is a maximal bi-colored subset of edges, including necessarily the color zero attributing to the Wick contractions. The set builds a cycle, which can be closed or open, respectively for closed and open faces. The set of zero-colored edges along the cycle define the boundary $\partial f$ of the face $f$, and its length is equal to the number of internal zero colored edges along with it.
\end{definition}
We moreover recall that the set of nodes hooked with an external edge are said \textit{external nodes}, and the vertices hooked to them are \textit{external vertices}. Figure \eqref{Feynman1} provides an example of a Feynman diagram, where propagators are materialized with dotted edges. A specificity of tensorial group field theory is that they are power countable. The most divergent diagrams define the melonic sector, and melonic diagrams can be defined recursively as for melonic bubbles (definition \ref{defmelons}), including edges of color zero in the recursion. They moreover satisfy the following lemma:
\begin{lemma}
Let $\mathcal{G}$ be a melonic 1PI $2N$-points diagram with more than one vertex. The external black and white nodes allows to build $N(d-1)$-dipoles. There are moreover $N$ external faces of the same color, whose boundaries connect to pairwise external nodes of different external $(d-1)$-dipoles. \label{propmelonfaces}
\end{lemma}
\begin{figure}
\begin{center}
\includegraphics[scale=1.2]{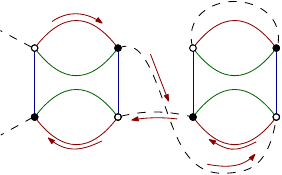}
\end{center}
\caption{A typical Feynman graph with two vertices, two external edges, two external nodes but one external vertex. The red arrows follow the boundary of a external face of length $2$.}\label{Feynman1}
\end{figure}
A proof of this lemma is given in Appendix \ref{App1} for the quartic melonic model. Before move on the power counting and specify the relevant models for this paper, we add two important definitions, the \textit{sum} of two bubbles and the \textit{boundary} of a Feynman graph:
\begin{definition}
Let $b_2^{(i)}$ and $b_2^{(j)}$ two bubbles. Let $n\in b_2^{(i)}$ and $\bar{n}\in b_2^{(j)}$ two white and black nodes and $e$ a dotted edge joining together with these two nodes. The connected sum $b_2^{(i)}\sharp_{n\bar{n}}b_2^{(j)}$ is defined as the bubble obtained from the two following successive moves:
\begin{enumerate}
\item Deleting the edge $e$
\item Connecting together the colored edges hooked to $n$ and $\bar{n}$ following their respective colors
\end{enumerate}
\end{definition}\label{defsum}
\begin{definition}
Let $\mathcal{G}$ be a connected Feynman graph with amputated of its external edges. Let $\partial n$ the set of external nodes (i.e. nodes hooked to external edges), $F^\prime=\{f^\prime\}$ the set of external faces and $\partial f^\prime$ the boundary of $f^\prime$. The boundary graph $\partial \mathcal{G}$ of $\mathcal{G}$ is the $d$-colored bipartite regular graph build as follows:
\begin{enumerate}
\item set $(n,\bar{n})\in \partial n$ two boundary nodes respectively black and white.
\item set $\partial_{n\bar{n}} F^\prime\subset \partial F^\prime$ the subset of boundaries having the couple $(n,\bar{n})$ as boundaries.
\item For each path $\partial f^\prime\in \partial_{n\bar{n}} F^\prime$ we create an edge of color $c(f^\prime)$ between $n$ and $\bar{n}$.
\item and so one for each pair $(n,\bar{n})\in \partial n$ .
\end{enumerate}
\end{definition}\label{defboundary}
Divergences can occur and require renormalization to well define the theory. For a compact Lie group, the divergent degree $\omega(\mathcal{G})$ for the graph $\mathcal{G}$ is uniformly bounded by the power counting \cite{Carrozza:2016tih}-\cite{Lahoche:2019cxt}.
\begin{theorem}
The Abelian power counting $\omega(\mathcal{G})$ for the graph $\mathcal{G}$ is given by:
\begin{equation}
\omega(\mathcal{G})=-2L(\mathcal{G})+D(F(\mathcal{G})-R(\mathcal{G}))\,,
\end{equation}
where $D$ is the dimension of the group $\mathrm{G}$, and $R$ the rank of the adjacency matrix $\epsilon_{ef}$, such that $\epsilon_{ef}=\pm 1$ for $e\in \partial f$ and $0$ otherwise; the sign depending on the relative orientation of the face\footnote{Because the graph of a complex model is orientable, we can fix the sign to be $1$ everywhere.}.
\end{theorem}
Note that the rank $R$ comes from the Gauss constraint, which improves the power counting for an unconstrained model. We have the following theorem (see \cite{Samary:2012bw} for more detail).
\begin{theorem}
The Abelian quartic melonic model with closure constraint in dimension $d=6$ is just renormalizable.
\end{theorem}\label{th1}

The model we consider in this note is the $\varphi^4$ model in dimension $d=6$ i.e. defined in the $6$-copy of the group $\U(1)$. Then the field can be decomposed along the Fourier modes $\{e^{i\sum_{j=1}^d p_j \theta_j}\}$, for $\theta_j\in [0,2\pi[$ and $p_j\in \mathbb{Z}$:
\begin{equation}
\varphi(\textbf{g})= \sum_{\textbf{p}\in \mathbb{Z}^6}\, T_{\textbf{p}} \, e^{i \textbf{p}\cdot \mathbf{\theta}}\,, \qquad \bar{\varphi}(\textbf{g})= \sum_{\textbf{p}\in \mathbb{Z}^6}\, \bar{T}_{\textbf{p}} \, e^{-i \textbf{p}\cdot \mathbf{\theta}}\,,
\end{equation}
where $\textbf{p}=(p_1,\cdots,p_d)$, $\mathbf{\theta}=(\theta_1,\cdots,\theta_d)$ and $\textbf{p}\cdot \mathbf{\theta}:=\sum_j p_j\theta_j$. It is suitable to work with Fourier components $T_{\textbf{p}}$ (resp. $\bar{T}_{\textbf{p}}$) rather than original fields $\varphi$ (resp. $\bar{\varphi}$). Under these considerations, the bare propagator reads as a diagonal matrix $C(\textbf{p},\textbf{p}^\prime)=\delta_{\textbf{p},\textbf{p}^\prime} C(\textbf{p})$, with
\begin{equation}
C(\textbf{p})=\frac{\delta\left(\sum_j p_j \right)}{Z_{-\infty}\textbf{p}^2+Z_2m^2_0},
\end{equation}
the delta ensuring closure constraint; and the interaction part of the action being:
\begin{equation}
S_{\text{int}}[T,\bar{T}]=Z_4\lambda_{4}\,\sum_{i=1}^6 \vcenter{\hbox{\includegraphics[scale=0.8]{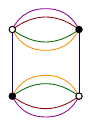} }}\,.\label{classicAction2}
\end{equation}
In these equations, $Z_{-\infty}$, $Z_2$, and $Z_4$ denote respectively the wave function, mass and couplings counter-terms. The renormalizability theorem \ref{th1} states that there exists a (not unique) choice of such a counter-terms which cancel all the high-momenta divergences arising in the perturbative expansion, up to a suitable regularization of large momenta sums. Note also that the index “$-\infty$" refers to the fact that the finite part of the counter-terms is fixed in the deep infrared limit through some \textit{renormalization conditions} and we provide this explicitly in section \ref{secEVE}. Besides this theorem provides no indications on the behavior of the RG flow and the underlying effective physics. This is why we move on to the nonperturbative technique as the Wetterich formalism.
\medskip

\subsection{Functional renormalization group for TGFTs}\label{sec2.3}
The renormalization group flow describes evolution of couplings through the effective action when UV degrees of freedom are integrated out. The FRG formalism is the most powerful tool to address nonperturbative issues. The formalism is aiming to construct a smooth interpolation, depending on a unique scale parameter $k$, between a microscopic action in the deep UV ($k=\Lambda$) to an effective action $\Gamma$ in the deep IR ($k=0$). For a just-renormalizable models, the microscopic scale $\Lambda$ becomes irrelevant, and may be removed in the continuum limit $\Lambda\to\infty$. For $k\in]0,\Lambda[$, between boundaries, we deal with an effective object $\Gamma_k$, which is nothing but the effective action at scale $k$ for the integrated-out degrees of freedom. It is defined as a slightly modified Legendre transform:
\begin{equation}
\Gamma_k[M,\bar{M}]+R_k[M,\bar{M}]=\sum_\textbf{p}\bar{\J}_\textbf{p} M_\textbf{p}+\sum_\textbf{p}\bar{M}_\textbf{p} \J_\textbf{p}-W_k[\J,\bar{\J}]\,,\label{average}
\end{equation}
where we introduced:
\medskip

\noindent
$\bullet$ $W_{k}[\J,\bar{\J}]$ as the free energy of the integrated degrees of freedom; formally given by the logarithm of the partition function:
\begin{equation}
Z_k[\J, \bar{\J}]= e^{\sum_{\textbf{p}} \, \frac{\partial}{\partial T_\textbf{p}} C(\textbf{p})\frac{\partial}{\partial \bar{T}_\textbf{p}}}\, \exp \bigg(-S_{\text{int}}[T,\bar{T}]-R_k[T,\bar{T}]+\sum_\textbf{p}\bar{\mathrm{J}}_\textbf{p}T_\textbf{p}+\sum_\textbf{p} \bar{T}_\textbf{p}{\mathrm{J}}_\textbf{p}\bigg)\bigg\vert_{T=\bar{T}=0}\,.\label{partitionWett}
\end{equation}
\medskip

\noindent
$\bullet$ The classical fields $M$ and $\bar{M}$ are themselves tensor fields, defined as:
\begin{equation}
M_{\textbf{p}}:= \frac{\partial W_k}{\partial \bar{\J}_{\textbf{p}}}\,,\qquad \bar{M}_{\textbf{p}}:= \frac{\partial W_k}{\partial {\J}_{\textbf{p}}}\,.\label{M}
\end{equation}
\medskip

\noindent
$\bullet$ The $R_k$-term in the definition \eqref{average} adds like a scale dependent mass, which decouple the low energy modes ($\textbf{p}^2<k^2$) from long distance physics whereas large energy modes remains less massive. Physically, the regulator ensures that the following boundary conditions hold:
\begin{equation}
\Gamma_{k>\Lambda}\to S\,,\qquad \Gamma_{k=0}=\Gamma\,.\label{boundary}
\end{equation}
Explicitly, $R_k$ is given in the Fourier modes as:
\begin{equation}
R_k[T,\bar{T}]:=\sum_\textbf{p} \bar{T}_\textbf{p} \,r_k(\textbf{p}^2){T}_\textbf{p}\,.
\end{equation}
The function $r_k(\textbf{p}^2)$ being positive defined, and having to satisfy some requirement in regard to the boundary conditions \eqref{boundary}, among them:
\begin{enumerate}
\item $\lim_{k\to\Lambda}r_k(\textbf{p}^2)\gg 1\,,$
\item $\lim_{k\to0}r_k(\textbf{p}^2) =0\,,$
\item $r_k(\textbf{p}^2>k^2)\simeq 0\,.$
\end{enumerate}
Because of the definition \eqref{average}, the effective $2$-point function $G_k$ have to be related with the second derivative of the classical action $\Gamma^{(2)}_k:=\partial_M\partial_{\bar{M}}\Gamma_k$ as:
\begin{equation}
G_k(\textbf{p},\textbf{p}^\prime):=\left(\Gamma^{(2)}_k+r_k\mathbb{I}\right)^{-1}(\textbf{p},\textbf{p}^\prime)\,,\label{twopoint}
\end{equation}
The effective average action moves through the theory space with the running scale $k$, and its evolution obeys the exact flow equations known as  \textit{Wetterich-Morris equation} \cite{Wetterich:1991be}-\cite{Morris:2000hm}:
\begin{equation}
\dot{\Gamma}_k=\sum_{\textbf{p}\in\mathbb{Z}^6} \dot{r}_k(\textbf{p}^2)\left(\Gamma^{(2)}_k+r_k\mathbb{I}\right)^{-1}(\textbf{p},\textbf{p})\,,\label{Wetterich}
\end{equation}
where the dot designates the derivative with respect to the normal coordinate along with the flow: $s:=\ln k/\Lambda$. Despite its apparent simplicity,  equation \eqref{Wetterich} works in an infinite-dimensional functional space and is hard to solve even for simple models. We generally use approximations for any nonperturbative information about the RG flow. Note that the Abelian model that we consider has been already studied in \cite{Lahoche:2016xiq}, using the crude truncations methods, which consist to imposes $\Gamma^{(2n)}_k=0$ for $n$ large enough. The conclusion in the large $k$ regime was the existence of a non-perturbative Wilson-Fisher fixed point, reminiscent of a second-order phase transition scenario. In this paper, we construct a new version of the EVE applicable to the model defined with the closure constraint. Note that the EVE technique is developed through a recent series of papers \cite{Lahoche:2018ggd}-\cite{Lahoche:2019cxt}, allowing to truncate 'smoothly' in the theory space around the marginal sector; and to capture the full melonic sector through a closing technique. This will be the topic of the Section \ref{secEVE} below.
\medskip

To close this section, let us add a remark about the notion of \textit{scaling dimension}. It is well known in standard quantum field theories (QFTs), that dimensions of the interactions are closely related to their power-counting renormalizability. The interactions with positive or vanishing (momentum) dimensions are renormalizable, while the ones with negative dimensions are non-renormalizable. For GFTs, however, there are no references scale \textit{a priori} in the classical action \eqref{classicAction2}, and the sums over $\mathbb{Z}^d$ are dimensionless. There are two ways to introduce a dimension in the QFT framework and recover the standard classification following the dimension of the couplings. The first one is to make contact with physical quantities. In particular, for non-Abelian models over $\SU(2)$, the spectrum of the Laplacian is as well the spectrum of the area operator of the LQG, providing it with a dimension $2$. The second strategy is to define the dimension from the behavior of the renormalization group flow itself. Indeed, for a just-renormalizable model, one expects the existence of a leading logarithm behavior to each order of the perturbative expansion. Thus, because they scale as $\Lambda^0$ with respect to some UV cut-off $\Lambda$, we attribute to them a scaling dimension zero. Next, it can be checked that any leading order $N$-points function made only with $4$-points melonic bubbles have to scale as $\Lambda^\omega$ (see \cite{Lahoche:2015ola} for more detail), with
\begin{equation}
\omega=4-N\, \label{powercountingTrue}
\end{equation}
This, for instance, defines the scaling dimension of the mass term, the corresponding quantum corrections coming from $N=2$ diagrams, the dimension of $m^2$ must be $2$. Moreover, all the couplings involving more than $4$ fields have a negative dimension, which is in agreement with our interpretation about a just-renormalizable theory (for an extended discussion, the reader may consult \cite{Lahoche:2018oeo}). \\

\section{Melonic RG solution}\label{secEVE}
The difficulty to solve the RG equation \eqref{Wetterich} can be recognized from the observation that; taking successive derivatives with respect to $M$ and $\bar{M}$; we can obtain the flow equations for $\Gamma^{(n)}_k$ in terms of $\Gamma^{(n+2)}_k$, $\Gamma^{(n)}_k$ and so on. Then, the solution of the flow equation \eqref{Wetterich} requires in principle to solve an infinite hierarchy; with an infinite number of initial conditions. Then some approximations are required to extract any nonperturbative information on the true behavior of the RG flow. An approximate solution can be constructed by projection into a suitable subspace. Generally, this subspace is spanned with a finite number of couplings; and focus on a finite number of interactions, and thus to impose $\Gamma^{(n)}_k=0$ for $n$ large enough. The originality of the EVE approach is to truncate around a sector, without cutting on $n$; but spanned with relevant and marginal couplings only. Moreover in this approach, effective vertices are no longer cut around zero momenta, and the full momentum dependence is kept. Using this non-trivial breakthrough, the singularity line arising from the computation of the anomalous dimension within a crude truncation disappears. The price to pay is that our investigations have to focus on the intermediate regime $\Lambda \gg k\gg 1$ called UV domain; far from UV regime to ensure the large $N$-limit; but far enough from the deep UV so that we have left the transient regime on which the flow is well parameterized by the relevant and marginal interactions through the \textit{large river effect} \cite{Delamotte:2007pf}. There are two additional internal approximations with our method. First, we focus on the \textit{symmetric phase}, defined as the region of the full phase space where an expansion around vanishing classical fields $M=\bar{M}=0$ holds. This simplification, in particular, ensures the following property for the effective $2$-point function \eqref{twopoint}:
\begin{property}\label{prop1}
In the symmetric phase, all the odd effective vertex functions, having not the same number of derivatives with respect to $M$ and $\bar{M}$ vanish. Moreover, the effective $2$-point function $G_k(\textbf{p},\textbf{p}^\prime)$ is diagonal:
\begin{equation}
G_k(\textbf{p},\textbf{p}^\prime):=G_k(\textbf{p})\delta_{\textbf{p}\textbf{p}^\prime}\,.
\end{equation}
\end{property}
See \cite{Lahoche:2018oeo}-\cite{Lahoche:2018vun} for more details. The previous statement ensures in particular that $\Gamma_k^{(n)}=0$ for $n=2p+1$. Moreover, note that because of the closure constraint, we must have $G_k(\textbf{p})\propto \delta(\sum_i p_i)$. In general, for effective vertex $\Gamma_k^{(2n)}$ we must have:
\begin{property}
In the symmetric phase, all the effective vertices $\Gamma_k^{(2n)}(\textbf{p}_1,\cdots, \textbf{p}_{2n})$ are such that $\sum_{i=1}^d p_{\ell i}=0\,\, \forall \ell \in \llbracket 1, 2n \rrbracket$\,.
\end{property}
This property comes straightforwardly from the observation that each Wick contraction with the propagator \eqref{propaprojection} project ‘‘on the right and the left". Hence, because external fields are necessarily contracted with a propagator, their arguments are \textit{de facto} constraint to have a vanishing sum \cite{Benedetti:2015yaa}.
\medskip

The second approximation is that we work on the non-branching sector defined in \ref{defnonbranch}. This can be justified from the observation that all the effective vertices can be labeled with a bubble corresponding to the pattern of contractions of external momenta. It corresponds to the boundary graph in the sense of definition \ref{defboundary} of any Feynman graph contributing to the perturbative expansion of such an effective vertex. In the melonic sector, all the boundaries are melonic bubbles following the definition \ref{defmelons}. Moreover, it can be checked recursively that, starting with effective vertices indexed with non-branching melonic bubbles, we can not generate leading order one-loop contributions outside of the non-branching subspace (our explicit computations in this section supporting this argument). Up to these approximations, EVE method aims to close the infinite hierarchy arising from the exact RG flow equation \eqref{Wetterich} using the recursive definition of melons. In the functional space, these take the form of non-trivial relations between effective vertices, all of them $\Gamma_k^{(2n)}$ for $n>2$ being expressed in terms of $\Gamma_k^{(4)}$ and $\Gamma_k^{(2)}$ only, closing the hierarchy in around the marginal interaction without sharp restriction on the momentum dependence of vertices and their number. Only the initial conditions are imposed, the flow being initialized as close as possible to the renormalizable sector. Now let us defined the symmetry invariance of our model through the WT identities.

\subsection{Ward-Takahashi identities}

In this section, we derive the WT identities for our model defined  in the relation \eqref{partition1}.  As we explained above (see section \eqref{sec21}), due to the gauge-invariant condition (the closure constraint) the ordinary method to derive these identities  
\cite{Samary:2014oya}-\cite{Ellwanger:1994iz} see also \cite{Wetterich:1992yh}-\cite{Wetterich:2017aoy},  need to be revisited with  great care. Our starting point is to rewrite the equivalent relations of \eqref{eq3} in the dual space $\mathbb{Z}^{D}$ and work with the Fourier components $T_{\textbf{p}}$ and $\bar{T}_{\textbf{p}}$ as:
\beq
T^{\prime}_{\textbf{p}}=[\mathcal{U}T]_{p_{1}...p_{d}}=\sum_{\textbf{q}}\Big[\prod^{d}_{i=1}U_{i, p_{i}q_{i}}\Big]T_{q_{1}, ..., q_{d}}
\eeq
\beq
\bar{T}^{\prime}_{\textbf{p}}=[T\mathcal{U}]^{\dagger}_{p_{1}...p_{d}}=\sum_{\textbf{q}}\Big[\prod^{d}_{i=1}\bar{U}_{i, p_{i}q_{i}}\Big]\bar{T}_{q_{1}, ..., q_{d}}
\eeq
where $\textbf{q}=(q_{1}, ..., q_{d})$. For tensorial group field theory, the interaction terms in the action are built such that it remains invariant under the unitary symmetry i.e. $[\mathcal{U}S_{int}]=S_{int}$. This is not the case for the kinetic term, due to the nontrivial propagator\footnote{Note that in the case of trivial propagator the kinetic action is invariant only in the initial condition on the flow i.e. in the UV regime.  This invariance is broken due to the appearance of the regulator.}. We have the following statement:
Instead of considering the global transformation on all the indices of the tensor, we can just work with a particular index $i$ and the final result will be, such that the variation of the action becomes:
\bea\label{vari2}
\delta^{\otimes U}S=\sum_{i=1}^d\delta^{(i)} S,
\eea
where $\delta^{\otimes U}S$ is the variation of the action under the global transformation $U(N)\otimes U(N)\cdots\otimes U(N):=U(N)^{\otimes d}$ and $\delta^{(i)} S$ is the variation of the action under U(N) transformation acting on the index $i$ of the tensor. To be more precise,
note that this property is satisfied only if we consider the infinitesimal first order $U(N)$ transformation.
Remark that by considering the $m$ order Taylor expansion of $U_{i, p_{i}p^{\prime}_{i}}$ we can generalize the relation \eqref{vari2} as
\bea\label{vari3}
\delta^{\otimes U}S=\sum_{i_1\neq i_2\neq\cdots \neq i_m=1}^d\delta^{(i_1,\cdots,i_m)} S.
\eea
The starting point to derive the Ward identities is to provide the variation of the partition function under this unitary transformation. Because the inverse of the propagator is not well defined, we use the following identity
\bea
Z_k[\J,\bar \J]=e^{\sum_{\textbf{p}}\frac{\partial}{\partial T_{\textbf{p}}}C(\textbf{p})\frac{\partial}{\partial\bar{T}_{\textbf{p}}}}e^{-S_{int}[T, \bar{T}]-R_k[T, \bar{T}]+\bar{T}\cdot\J+\bar{\J}\cdot T]}\,\,\Big|_{T=\bar{T}=0},
\eea
where $\bar{T}\cdot \J+\bar{\J}\cdot T=\sum_{\textbf{p}}\big(\bar{T}_{\textbf{p}}\J_{\textbf{p}}+\bar{\J}_{\textbf{p}}T_{\textbf{p}}\big)$.  We shall focus on the first index $i=1$ and we denote by $U_{1}$ the transformation acting on the fields $T$ and $(\bar{T}, resp.)$. From the previous relations, we get the partial derivative with respect to tensor components transformation laws given by:
\bea
\frac{\partial}{\partial T^{\prime}_{\textbf{p}}}=\sum_{p^\prime_{1}}U^{-1}_{p^\prime_{1}p_{1}}\frac{\partial}{\partial T_{p^\prime_{1}\textbf{p}_{\perp}}}\,\,\,\,,\,\,\,\,
\frac{\partial}{\partial\bar{T}^{\prime}_{\textbf{p}}}=\sum_{p^\prime_{1}}\bar{U}^{-1}_{p^\prime_{1}p_{1}}\frac{\partial}{\partial \bar{T}_{p^\prime_{1}\textbf{p}_{\perp}}},
\eea
where we define $\textbf{p}_{\perp}=(0, p_{2}, ..., p_{d})$.
The $U_{1}$ infinitesimal transformation can be written as:
\begin{equation}
U_{1, p_{1}p^{\prime}_{1}}\simeq\delta_{p_{1}p^{\prime}_{1}}+\epsilon_{1, p_{1}p^{\prime}_{1}}+O(\epsilon^{2}_{1}),\,\,\,\,\,\, \bar{U}_{1, p_{1}p^{\prime}_{1}}\simeq\delta_{p_{1}p^{\prime}_{1}}+\bar{\epsilon}_{1, p_{1}p^{\prime}_{1}}+O(\bar{\epsilon}^{2}_{1}),
\end{equation}
where $\epsilon$ is infinitesimal anti-hermitian operator corresponding to the generator of unitary group element $U_{1}$, then, $\epsilon^{\dagger}_{1}=-\epsilon_{1}$. Acting $U_{1}$ on the first tensor index, allow us get the following transformations:
\bea
&&\frac{\partial}{\partial T}C\frac{\partial}{\partial\bar{T}}\rightarrow\frac{\partial}{\partial T}C\frac{\partial}{\partial\bar{T}}+\sum_{\textbf{p},\textbf{p}^{\prime}}\prod_{j\neq1}\delta_{p_{j}p^{\prime}_{j}}\Big[\frac{\partial}{\partial T_{\textbf{p}}}\big[C(\textbf{p})-C(\textbf{p}\,^\prime)\big]\frac{\partial}{\partial\bar{T}_{\textbf{p}\,^\prime}}\Big]\epsilon_{1, p_{1}p^{\prime}_{1}}
\\
&& R_{k}\rightarrow R_{k}+\sum_{\textbf{p},\textbf{p}\,^\prime}\prod_{j\neq1}\delta_{p_{j}p^\prime_{j}}T_{\textbf{p}\,^\prime}\big[r_{k}(\textbf{p}^{2})-r_{k}(\textbf{p}^{\prime2})\big]\bar{T}_{\textbf{p}}\epsilon_{1, p_{1}p^\prime_{1}}
\\
&&F\rightarrow F+\sum_{\textbf{p},\textbf{p}^{\prime}}\prod_{j\neq1}\delta_{p_{j}p^{\prime}_{j}}\big[\bar{\J}_{\textbf{p}}T_{\textbf{p}^\prime}-\bar{T}_{\textbf{p}}\J_{\textbf{p}^\prime}\big]\epsilon_{1, p_{1}p^{\prime}_{1}}
\eea
where we adopt the following notations:
\bea
\sum_{\textbf{p}}\frac{\partial}{\partial T_{\textbf{p}}}C(\textbf{p})\frac{\partial}{\partial\bar{T}_{\textbf{p}}}=\frac{\partial}{\partial T}C\frac{\partial}{\partial\bar{T}},\quad F=\bar{T}\cdot \J+\bar{\J}\cdot T, \,R_{k}[T, \bar{T}]=R_{k}.
\eea
Now the total variation of the partition function under $\mathcal U$ transformations in the first order of the parameter $\epsilon$ becomes
\bea\label{nato}
&&\delta Z_k=\sum_{i=1}^d\exp\Big(\frac{\partial}{\partial T}C\frac{\partial}{\partial\bar{T}}\Big)\Bigg\{\sum_{\textbf{p},\textbf{p}^\prime}\prod_{j\neq i}\delta_{p_jp_j'}\bigg[\frac{\partial}{\partial T_{\textbf{p}}}\Delta C(\textbf{p}, \textbf{p}^\prime)\frac{\partial}{\partial\bar{T}_{\textbf{p}^\prime}}+\bar{T}_{\textbf{p}^{\prime}}\big(r_{k}(\textbf{p}^{2})-r_{k}(\textbf{p}^{\prime2})\big)\bar{T}_{\textbf{p}}\cr
&&\big(\bar{\J}_{\textbf{p}}T_{\textbf{p}^\prime}-\bar{T}_{\textbf{p}}\J_{\textbf{p}^\prime}\big)\bigg]\epsilon_{i, p_{i}p^{\prime}_{i}}\Bigg\}e^{-S_{int}-R_{k}+F}\,\,\Big|_{T=\bar{T}=0}=0
\eea
where $\Delta C(\textbf{p}, \textbf{p}^\prime)= C(\textbf{p})-C(\textbf{p}^\prime)$.
The identity given in \eqref{nato} is independant to the generator $\epsilon$ of the $U(N)$ transformation and therefore is independant for all indices $i=1,2,\cdots d$. Then by setting $i=1$ we get the following relation:
\bea\label{eq39}
&&\exp\Big(\frac{\partial}{\partial T}C\frac{\partial}{\partial\bar{T}}\Big)\sum_{\textbf{p}_{\perp} \textbf{p}^\prime_{\perp}}\prod_{j\neq 1}\delta_{p_jp_j'}\Big[T_{\textbf{p}^\prime}\bar{\J}_{\textbf{p}}-\bar{T}_{\textbf{p}}\J_{\textbf{p}^\prime}\Big]e^{-S_{int}-R_{k}+F}\,\,\Big|_{T=\bar{T}=0}=0.
\eea
Note that in the previous relation we have deleted the quantities $\frac{\partial}{\partial T_{\textbf{p}}}\Delta C(\textbf{p}, \textbf{p}^\prime)\frac{\partial}{\partial\bar{T}_{\textbf{p}^\prime}}$ and $\bar{T}_{\textbf{p}^{\prime}}\big(r_{k}(\textbf{p}^{2})-r_{k}(\textbf{p}^{\prime2})\big)\bar{T}_{\textbf{p}}$ because of the closure constraint.
Then \eqref{eq39} can be simply re-expressed as  follows:
\bea
&&\sum_{\textbf{p}_{\perp} \textbf{p}^\prime_{\perp}}\prod_{j\neq 1}\delta_{\textbf{p} \textbf{p}^\prime}\Big[\bar \J_{\textbf{p}}\exp\Big(\frac{\partial}{\partial T}C\frac{\partial}{\partial\bar{T}}-S_{int}-R_{k}+F\Big)T_{\textbf{p}^\prime}\cr
&&-\J_{\textbf{p}^\prime}\exp\Big(\frac{\partial}{\partial T}C\frac{\partial}{\partial\bar{T}}-S_{int}-R_{k}+F\Big)\bar{T}_{\textbf{p}}\Big]\,\,\Big|_{T=\bar{T}=0}=0.
\eea
Finally after few algebraic computation, in terms of the generator of the connected Green function $W_k$ we get the relation:
\beq
\sum_{\textbf{p}_{\perp} \textbf{p}^\prime_{\perp}}\delta_{\textbf{p}_{\perp} \textbf{p}^\prime_{\perp}}\bigg[\bar{\J}_{\textbf{p}}\frac{\partial}{\partial\bar{\J}_{\textbf{p}^\prime}}-\J_{\textbf{p}^\prime}\frac{\partial}{\partial \J_{\textbf{p}}}\bigg]e^{W_{k}}=0,
\eeq
and the following result holds:
\begin{theorem}\label{WardID1}
The first order WT identities of the model  defined by \eqref{partition1} are given by the following relation
\beq
\mI_i=\sum_{ \textbf{p}_{\bot i}\textbf{p}\,^\prime_{\bot i}}\delta_{\textbf{p}_{\perp i} \textbf{p}^\prime_{\perp i}}\Big[\bar{\J}_{\textbf{p}}M_{\textbf{p}^\prime}-\bar{M}_{\textbf{p}}\J_{\textbf{p}^\prime}\Big]=0,\label{socle2}
\eeq
where $M_{\textbf{p}^\prime}=\frac{\partial W_{k}}{\partial \bar{\J}_{\textbf{p}^\prime}}, \,\,\,\, \bar{M}_{\textbf{p}}=\frac{\partial W_{k}}{\partial\J_{\textbf{p}}}$ are the mean fields.
\end{theorem}
Let us remark that the above identity given in the relation \eqref{socle2} is irrelevant in our analysis i.e., this identity is trivial and can not contribute as a new constraint in the Wetterich flow equation. To be more precise, computing the first-order derivative with respect to $\J$ and $\bar \J$ of \eqref{socle2} leads to
\bea
G_{k}(\textbf{p}, \textbf{p}^\prime)=G_k(\textbf{p})\delta_{\textbf{p} \textbf{p}^\prime},
\eea
which corresponds to the symmetric phase condition. Now let us investigate the same computation by taking into account the second order expansion of the $U(N)$ transformation as:
\bea
&& U_{i, p_{i}p^{\prime}_{i}}\simeq\delta_{p_{i}p^{\prime}_{i}}+\epsilon_{i, p_{i}p^{\prime}_{i}}+\frac{1}{2}\sum_q \epsilon_{i, p_{i}q}\epsilon_{i, qp^{\prime}_{i}} +O(\epsilon^{3}_{i}),\\ &&\bar{U}_{i, p_{i}p^{\prime}_{i}}\simeq\delta_{p_{i}p^{\prime}_{i}}+\bar{\epsilon}_{i, p_{i}p^{\prime}_{i}}+\frac{1}{2}\sum_q\bar\epsilon_{i, p_{i}q}\bar\epsilon_{i, qp^{\prime}_{i}} +O(\epsilon^{3}_{i}).
\eea
The full variation of the action becomes
$
\delta^{\otimes U}S=\sum_{i\neq j}^d\delta^{(i,j)} S.
$
Then, by setting $(i,j)=(1,2)$, the variation of the partition function may be written as
\bea
\delta Z_k=0 \Leftrightarrow \mI_1\epsilon_1+\mI_2\epsilon_2+ \mI_{12}\epsilon_1\epsilon_2+ \mI_{11}\epsilon_1\epsilon_1+ \mI_{22}\epsilon_2\epsilon_2=0,
\eea
where the first order WT identity given in \eqref{socle2}  is
well satisfied and corresponds to $\mI_1=0$ and $\mI_2=0$. By taking into account only the second order terms on $\epsilon$, we obtain new identities that we denote by second order WT identities, i.e. $\mI_{12}=0$, $\mI_{11}=0$, $\mI_{22}=0$, where

\bea
&& \mI_{12}=\sum_{\textbf{p}_{\bot 1}, \textbf{p}'_{\bot 1}}\sum_{\textbf{q}_{\bot 2}, \textbf{q}'_{\bot 2}}\delta_{\textbf{p}_{\bot 1}\textbf{p}'_{\bot 1}}\delta_{\textbf{q}_{\bot 2}\textbf{q}'_{\bot 2}}\Big(\bar{\J}_{\textbf{p}}\bar{\J}_{\textbf{q}}G^{(0, 2)}_{k}(\textbf{q}^\prime; \textbf{p}^\prime)
+M_{\textbf{p}^\prime}M_{\textbf{q}^\prime}\bar{\J}_{\textbf{p}}\bar{\J}_{\textbf{q}}-\bar{\J}_{\textbf{p}}\J_{\textbf{q}^\prime}G^{(1, 1)}_{k}(\textbf{q}; \textbf{p}^\prime)\cr&&-\bar{\J}_{\textbf{p}}\J_{\textbf{q}^\prime}\bar
M_{\textbf{q}}M_{\textbf{p}^\prime}-\J_{\textbf{p}^\prime}\bar{\J}_{\textbf{q}}G^{(1, 1)}_{k}(\textbf{p}; \textbf{q}^\prime)-\bar{M}_{\textbf{p}}M_{\textbf{q}^\prime}\J_{\textbf{p}^\prime}\bar{\J}_{\textbf{q}}+\J_{\textbf{p}^\prime}\J_{\textbf{q}^\prime}G^{(2, 0)}_{k}(\textbf{p}; \textbf{q})+\J_{\textbf{p}^\prime}\J_{\textbf{q}^\prime}\bar{M}_{\textbf{p}}\bar{M}_{\textbf{q}}\Big)\cr
&& +\sum_{\textbf{p}_{\perp\perp}, \textbf{p}^\prime_{\perp\perp}}\delta_{\textbf{p}_{\perp\perp}\textbf{p}^\prime_{\perp\perp}}\big(\bar{\J}_{\textbf{p}}M_{\textbf{p}^\prime}+\bar{M}_{\textbf{p}}\J_{\textbf{p}^\prime}\big)\delta_{{p}'_{2}q'_{2}}\delta_{{p}_{2}q_{2}}
\eea
\bea
&& \mI_{11}=\sum_{\textbf{p}_{\bot 1}, \textbf{p}'_{\bot 1}}\sum_{\textbf{q}_{\bot 1}, \textbf{q}'_{\bot 1}}\delta_{\textbf{p}_{\perp 1}\textbf{p}^\prime_{\perp 1}}\delta_{\textbf{q}_{\perp 1}\textbf{q}^{\prime}_{\perp 1}}\Big(\bar{\J}_{\textbf{p}}\bar{\J}_{\textbf{q}}G^{(0, 2)}_{k}(\textbf{q}^\prime; \textbf{p}^\prime)
+M_{\textbf{p}^\prime}M_{\textbf{q}^\prime}\bar{\J}_{\textbf{p}}\bar{\J}_{\textbf{q}}-\bar{\J}_{\textbf{p}}\J_{\textbf{q}^\prime}G^{(1, 1)}_{k}(\textbf{q}; \textbf{p}^\prime)\cr&&-\bar{\J}_{\textbf{p}}\J_{\textbf{q}^\prime}\bar
M_{\textbf{q}}M_{\textbf{p}^\prime}-\J_{\textbf{p}^\prime}\bar{\J}_{\textbf{q}}G^{(1, 1)}_{k}(\textbf{p}; \textbf{q}^\prime)-\bar{M}_{\textbf{p}}M_{\textbf{q}^\prime}\J_{\textbf{p}^\prime}\bar{\J}_{\textbf{q}}+\J_{\textbf{p}^\prime}\J_{\textbf{q}^\prime}G^{(2, 0)}_{k}(\textbf{p}; \textbf{q})+\J_{\textbf{p}^\prime}\J_{\textbf{q}^\prime}\bar{M}_{\textbf{p}}\bar{M}_{\textbf{q}}\Big)\cr
&&+\sum_{\textbf{p}_{\perp 1}, \textbf{p}'_{\perp 1}}\delta_{\textbf{p}_{\perp 1}\textbf{p}^\prime_{\perp 1}}\big(\bar{\J}_{\textbf{p}}M_{\textbf{p}^{\prime}}\delta_{p_1q_1}\delta_{p_1q'_1}+\bar{M}_{\textbf{p}}\J_{\textbf{p}^{\prime}}\delta_{p'_1q_1}\delta_{p'_1q'_1}\big)
\eea
\bea
&& \mI_{22}=\sum_{\textbf{p}_{\bot 2}, \textbf{p}'_{\bot 2}}\sum_{\textbf{q}_{\bot 2}, \textbf{q}'_{\bot 2}}\delta_{\textbf{p}_{\perp 2}\textbf{p}^\prime_{\perp 2}}\delta_{\textbf{q}_{\perp 2}\textbf{q}^{\prime}_{\perp 2}}\Big(\bar{\J}_{\textbf{p}}\bar{\J}_{\textbf{q}}G^{(0, 2)}_{k}(\textbf{q}^\prime; \textbf{p}^\prime)
+M_{\textbf{p}^\prime}M_{\textbf{q}^\prime}\bar{\J}_{\textbf{p}}\bar{\J}_{\textbf{q}}-\bar{\J}_{\textbf{p}}\J_{\textbf{q}^\prime}G^{(1, 1)}_{k}(\textbf{q}; \textbf{p}^\prime)\cr&&-\bar{\J}_{\textbf{p}}\J_{\textbf{q}^\prime}\bar
M_{\textbf{q}}M_{\textbf{p}^\prime}-\J_{\textbf{p}^\prime}\bar{\J}_{\textbf{q}}G^{(1, 1)}_{k}(\textbf{p}; \textbf{q}^\prime)-\bar{M}_{\textbf{p}}M_{\textbf{q}^\prime}\J_{\textbf{p}^\prime}\bar{\J}_{\textbf{q}}+\J_{\textbf{p}^\prime}\J_{\textbf{q}^\prime}G^{(2, 0)}_{k}(\textbf{p}; \textbf{q})+\J_{\textbf{p}^\prime}\J_{\textbf{q}^\prime}\bar{M}_{\textbf{p}}\bar{M}_{\textbf{q}}\Big)\cr&&
+\sum_{\textbf{p}_{\perp 2}, \textbf{p}'_{\perp 2}}\delta_{\textbf{p}_{\perp 2}\textbf{p}^\prime_{\perp 2}}\big(\bar{\J}_{\textbf{p}}M_{\textbf{p}^{\prime}}\delta_{p_2q_2}\delta_{p_2q'_2}+\bar{M}_{\textbf{p}}\J_{\textbf{p}^{\prime}}\delta_{p'_2q_2}\delta_{p'_2q'_2}\big)
\eea
where $\textbf{p}_{\perp 1}=(0, p_{2}, ..., p_{d}),\,\textbf{p}_{\perp 2}=(p_{1}, 0, p_{3}, ..., p_{d}), \textbf{p}_{\bot\bot}=(0, 0,p_3, ..., p_{d}).$
\begin{proposition}\label{propo1}
The second order terms of the variation of the partition function do not provide an additional constraint on the correlation function i.e. $\mI_{ij},\, i=1,2$   vanish identically.
\end{proposition}
\begin{proof}
Before giving the proof of the proposition let us remark that this result is central for our investigation due to the fact that an additional constraint may allow being not compatible with the Wetterich flow equation. Recall that:
\beq\label{nanana}
M=\frac{\partial W}{\partial \bar{\J}} \mbox{ and } \bar M=\frac{\partial W}{\partial\J}\Rightarrow M':=\frac{\partial M}{\partial \J}=\frac{\partial \bar{M}}{\partial \bar\J}=:\bar{M}'.
\eeq
Let us consider the first order WT identity given by the relation \eqref{WardID1}. The partial derivative of this identity with respect to $\J$, in the first time, and with respect to $\bar \J$ in the second time are respectively
\bea
\frac{\partial \mI_i}{\partial \J_{\bf q}}&=&\sum_{ \textbf{p}_{\bot i}\textbf{p}\,^\prime_{\bot i}}\delta_{\textbf{p}_{\perp i} \textbf{p}^\prime_{\perp i}}\Big(\bar{\J}_{\bf p}\frac{\partial M_{\bf p'}}{\partial \J_{\bf q}}-\frac{\partial \bar{M}_{\bf p}}{\partial \J_{\bf q}}\J_{\bf p'}-\bar{M}_{\bf p}\delta_{\bf p'q}\Big)\cr
&=&\sum_{ \textbf{p}_{\bot i}\textbf{p}\,^\prime_{\bot i}}\delta_{\textbf{p}_{\perp i} \textbf{p}^\prime_{\perp i}}\Big(\bar{\J}_{\bf p}\frac{\partial^2 W_k}{\partial \J_{\bf q}\partial \bar{\J}_{\bf p'}}-\frac{\partial^2 W_k}{\partial \J_{\bf q}\partial \J_{\bf p}}\J_{\bf p'}-\bar{M}_{\bf p}\delta_{\bf p'q}\Big)=0
\eea
Multiplying this relation by $\J_{\bf q'}$ and using the fact that $\frac{\partial^2 W_k}{\partial \J_{\bf q}\partial \bar{\J}_{\bf p'}}=G_k^{(1,1)}({\bf q,p'})$, $\frac{\partial^2 W_k}{\partial \J_{\bf q}\partial \J_{\bf p}}=G_k^{(2,0)}({\bf p,q})$, we get:
\beq\label{vin12}
\sum_{ \textbf{p}_{\bot i}\textbf{p}\,^\prime_{\bot i}}\delta_{\textbf{p}_{\perp i} \textbf{p}^\prime_{\perp i}}\Big(G_k^{(2,0)}({\bf p,q})\J_{\bf q'}\J_{\bf p'}-G_k^{(1,1)}({\bf q,p'})\J_{\bf q'}\bar{\J}_{\bf p}\Big)=-\sum_{ \textbf{p}_{\bot i}\textbf{p}\,^\prime_{\bot i}}\delta_{\textbf{p}_{\perp i} \textbf{p}^\prime_{\perp i}}\delta_{\bf p'q} \J_{\bf q'}\bar{M}_{\bf p},
\eeq
and
\bea
\frac{\partial \mI_i}{\partial \bar{J}_{\bf q'}}&=&\sum_{ \textbf{p}_{\bot i}\textbf{p}\,^\prime_{\bot i}}\delta_{\textbf{p}_{\perp i} \textbf{p}^\prime_{\perp i}}\Big(\delta_{\bf q' p}M_{\bf p'}+\bar{\J}_{\bf p}\frac{\partial M_{\bf p'}}{\partial \bar{\J}_{\bf q'}}-\frac{\partial \bar{M}_{\bf p}}{\partial \bar{\J}_{\bf q'}}\J_{\bf p'}\Big)\cr
&=&\sum_{ \textbf{p}_{\bot i}\textbf{p}\,^\prime_{\bot i}}\delta_{\textbf{p}_{\perp i} \textbf{p}^\prime_{\perp i}}\Big(\delta_{\bf q' p}M_{\bf p'}+\bar{\J}_{\bf p}\frac{\partial^2 W_k}{\partial \bar{\J}_{\bf q'}\partial \bar{\J}_{\bf p'}}-\frac{\partial^2 W_k}{\partial \J_{\bf p}\partial \bar{\J}_{\bf q'}}\J_{\bf p'}\Big)
\eea
Multiplying this relation by $\bar{\J}_{\bf q}$ and using the fact that $\frac{\partial^2 W_k}{\partial \bar{\J}_{\bf q'}\partial \bar{\J}_{\bf p'}}=G^{(0,2)}_{k}({\bf p',q'})$, $\frac{\partial^2 W_k}{\partial \J_{\bf p}\partial \bar{\J}_{\bf q'}}=G^{(1,1)}_{k}({\bf p,q'})$ we get:
\beq\label{vin13}
\sum_{ \textbf{p}_{\bot i}\textbf{p}\,^\prime_{\bot i}}\delta_{\textbf{p}_{\perp i} \textbf{p}^\prime_{\perp i}}\Big(G^{(0,2)}_{k}({\bf p',q'})\bar{\J}_{\bf p}\bar{\J}_{\bf q}-G^{(1,1)}_{k}({\bf p,q'})\J_{\bf p'}\bar{\J}_{\bf q}\Big)=-\sum_{ \textbf{p}_{\bot i}\textbf{p}\,^\prime_{\bot i}}\delta_{\textbf{p}_{\perp i} \textbf{p}^\prime_{\perp i}}\delta_{\bf q' p}\bar{\J}_{\bf q}M_{\bf p'}.
\eeq
Now summing the expressions \eqref{vin12} and \eqref{vin13}, setting $i=1$ and by adding the sum $\sum_{ \textbf{q}_{\bot 2}\textbf{q}\,^\prime_{\bot 2}}\delta_{\textbf{q}_{\perp 2} \textbf{q}^\prime_{\perp 2}}$ we get:
\bea
&&\sum_{\textbf{p}_{\bot 1}, \textbf{p}'_{\bot 1}}\sum_{\textbf{q}_{\bot 2}, \textbf{q}'_{\bot 2}}\delta_{\textbf{p}_{\bot 1}\textbf{p}'_{\bot 1}}\delta_{\textbf{q}_{\bot 2}\textbf{q}'_{\bot 2}}\Big(\bar{\J}_{\textbf{p}}\bar{\J}_{\textbf{q}}G^{(0, 2)}_{k}(\textbf{q}^\prime; \textbf{p}^\prime)
-\bar{\J}_{\textbf{p}}\J_{\textbf{q}^\prime}G^{(1, 1)}_{k}(\textbf{q}; \textbf{p}^\prime)-\J_{\textbf{p}^\prime}\bar{\J}_{\textbf{q}}G^{(1, 1)}_{k}(\textbf{p}; \textbf{q}^\prime)\cr
&&+\J_{\textbf{p}^\prime}\J_{\textbf{q}^\prime}G^{(2, 0)}_{k}(\textbf{p}; \textbf{q})\Big)=-\sum_{\textbf{p}_{\bot 1}, \textbf{p}'_{\bot 1}}\sum_{\textbf{q}_{\bot 2}, \textbf{q}'_{\bot 2}}\delta_{\textbf{p}_{\bot 1}\textbf{p}'_{\bot 1}}\delta_{\textbf{q}_{\bot 2}\textbf{q}'_{\bot 2}}\Big(\delta_{\bf p'q} \J_{\bf q'}\bar{M}_{\bf p}+\delta_{\bf q' p}\bar{\J}_{\bf q}M_{\bf p'}\Big).
\eea
The right hand side of the above relation may be  simplified by computing the sum of the delta function which is written as
\bea
-\sum_{\textbf{p}_{\perp\perp}, \textbf{p}^\prime_{\perp\perp}}\delta_{\textbf{p}_{\perp\perp}\textbf{p}^\prime_{\perp\perp}}\big(\bar{\J}_{\textbf{p}}M_{\textbf{p}^\prime}+\bar{M}_{\textbf{p}}\J_{\textbf{p}^\prime}\big)\delta_{{p}'_{2}q'_{2}}\delta_{{p}_{2}q_{2}}.
\eea
Finally by using the first order WT identities $\mI_i=0$, we can simply deduce that
\begin{align}
&\sum_{\textbf{p}_{\bot 1}, \textbf{p}'_{\bot 1},\textbf{q}_{\bot 2}, \textbf{q}'_{\bot 2}}\delta_{\textbf{p}_{\bot 1}\textbf{p}'_{\bot 1}}\delta_{\textbf{q}_{\bot 2}\textbf{q}'_{\bot 2}}\Big(M_{\textbf{p}^\prime}M_{\textbf{q}^\prime}\bar{\J}_{\textbf{p}}\bar{\J}_{\textbf{q}}-\bar{\J}_{\textbf{p}}\J_{\textbf{q}^\prime}\bar
M_{\textbf{q}}M_{\textbf{p}^\prime}\cr
&-\bar{M}_{\textbf{p}}M_{\textbf{q}^\prime}\J_{\textbf{p}^\prime}\bar{\J}_{\textbf{q}}+\J_{\textbf{p}^\prime}J_{\textbf{q}^\prime}\bar{M}_{\textbf{p}}\bar{M}_{\textbf{q}}\Big)=0.
\end{align}
Therefore $\mI_{12}=0$,
 which shows that no additional information may be obtained by the second order WT identities. Note that the same conclusion holds by considering $\mI_{11}$ and $\mI_{22}$. To be more precise, let us remark that for $\mI_{11}$ the same computation holds up to expression \eqref{vin13}. Now summing the expressions \eqref{vin12} and \eqref{vin13}, setting $i=1$ and by adding the sum $\sum_{ \textbf{q}_{\bot 1}\textbf{q}\,^\prime_{\bot 1}}\delta_{\textbf{q}_{\perp 1} \textbf{q}^\prime_{\perp 1}}$, we get $\mI_{11}=0$. For $\mI_{22}$ the same computation holds up to expression \eqref{vin13}. Now summing the expressions \eqref{vin12} and \eqref{vin13}, setting $i=2$ and by adding the sum $\sum_{ \textbf{q}_{\bot 2}\textbf{q}\,^\prime_{\bot 2}}\delta_{\textbf{q}_{\perp 2} \textbf{q}^\prime_{\perp 2}}$, we get $\mI_{22}=0$, which ends the proof of the proposition.
\end{proof}
Remark that, due to the proposition \ref{propo1}, the higher order WT-identities are  the product of $\mI_i$ and $\mI_{k\ell}$  and therefore
vanish identically.
Finally, no WT-identity violation appears for the TGFT models with a closure constraint. In the next section, we derive the flow equations using the improved version of the EVE and provide the numerical analysis and the existence of Wilson-Fisher fixed point.

\subsection{Closing hierarchy in the non-branching sector}

Now let us focus on the melonic non-branching sector, defined in section \ref{sec21}, definition \ref{defnonbranch}, and on the symmetric phase. In such a sector of the theory space, the effective vertex functions $\Gamma_k^{n}$ are even $(n=2p)$, and decompose as a sum of functions indexed with a single color $i\in \llbracket 1,6 \rrbracket$,
\begin{equation}
\Gamma^{(2n)}_k(\{\textbf{p}_\ell\})=\sum_{i=1}^{d=6} \,\Gamma^{(b_n^{(i)})}_k(\{\textbf{p}_\ell\})\,,\label{decomp}
\end{equation}
where $b_n^{(i)}$ denote the non-branching bubble corresponding to the boundary graphs which leads to the Feynman amplitude involves in the perturbation theory. The non-branching components $\Gamma^{(b_n^{(i)})}_k$ have been extensively discussed in \cite{Lahoche:2018ggd}-\cite{Lahoche:2019cxt}; and we only provide the main statement in this section. The boundary graph index dictates the ways as the external momenta are identified. It materializes a product of Kronecker deltas, each colored edge in the equations \eqref{www} and \eqref{wwww} correspond to the bubble $b_n^{(i)}$ and then is associated to one of these deltas, following the definition of section \ref{sec21}. For our purpose, we essentially need to the $4$ and $6$-point bubbles, which reads explicitly:
\begin{align}
\label{www}
\vcenter{\hbox{\includegraphics[scale=0.9]{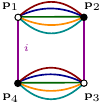} }}&\equiv\delta_{p_{1i}p_{4i}}\delta_{p_{2i}p_{3i}}\,\prod_{j\neq i}\delta_{p_{1j}p_{2j}}\delta_{p_{3j}p_{4j}}\,,\\
\label{wwww}\vcenter{\hbox{\includegraphics[scale=0.8]{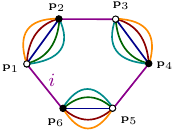} }}&\equiv\delta_{p_{1i}p_{6i}}\delta_{p_{2i}p_{3i}}\delta_{p_{4i}p_{5i}}\prod_{j\neq i}\delta_{p_{1j}p_{2j}}\delta_{p_{3j}p_{4j}}\delta_{p_{5j}p_{6j}}\,.
\end{align}
From these observations, the partial effective vertex functions $\Gamma^{(b_n^{(i)})}_k$ for a given melonic bubble $b_n^{(i)}$ reads as:
\begin{equation}
\Gamma^{(b_n^{(i)})}_k(\{\textbf{p}_\ell\})=:\sym\left(\pi_k^{(b_n^{(i)})}(\{p_{\ell\,i}\})\times\vcenter{\hbox{\includegraphics[scale=0.9]{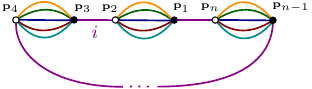} }}\right)\,, \label{vertexstructure}
\end{equation}
such that the symbol $\sym$ denotes the permutation of the external momenta. As an example:
\begin{equation}
\sym \left(\vcenter{\hbox{\includegraphics[scale=0.9]{Melon13N.pdf} }}\right)=2\left(\vcenter{\hbox{\includegraphics[scale=0.9]{Melon13N.pdf} }}+\vcenter{\hbox{\includegraphics[scale=0.9]{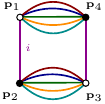} }}\right)\,,
\end{equation}
the factor $2$ in front of the definition arising from the fact that permuting both the black and white nodes do not change the global configuration. The functions $\pi_k^{(b_n^{(i)})}:\mathbb{Z}^n\to\mathbb{R}$ in equation \eqref{vertexstructure} depend only on the $i$-th components of the external momenta. Note that with this respect the gauge invariance imply a strong simplification. Indeed, let us consider a pair $(\textbf{p}_\ell,\textbf{p}_{\ell+1})$ of external momenta associated respectively with a black and a white node which allows to build a $(d-1)$-dipole. Because of the gauge invariance, the external momenta have to satisfy the closure constraints \cite{Benedetti:2015yaa}: $\sum_{j} p_{\ell j}=\sum_{j} p_{\ell+1 j}=0$. Due to the $d-1$ Kronecker delta between the two nodes, all the components with $j=i$ are identified: $p_{\ell j}=p_{\ell+1 j}$, $\forall \, j\neq i$, and finally $p_{\ell i}=p_{\ell+1 i}$. Because two $(d-1)$-dipoles are linked with a single Kronecker delta for a non-branching interaction, then all the components $p_{\ell,i}$ have to be the same, for each black or white node along the ring. The functions $\pi_k^{(b_n^{(i)})}$ depends on a single $i$-colored momenta $\pi_k^{(b_n^{(i)})}:\mathbb{Z}\to\mathbb{R}$ and the other colors being constraint to be equals. This is a very strong simplification  compared
to the unconstrained case, see \cite{Lahoche:2018oeo}. Finally:
\begin{equation}
\Gamma^{(b_n^{(i)})}_k(\{\textbf{p}_\ell\})\Big\vert_{\text{on-shell}}=:\pi_k^{(b_n^{(i)})}(p_{1i})\left(\prod_{\ell\neq i} \delta_{p_{1i}p_{\ell i}}\right)\times\sym\left(\vcenter{\hbox{\includegraphics[scale=0.9]{MelonNN.pdf} }}\right)\,. \label{vertexstructure2}
\end{equation}
\medskip

The zero-momenta values of the effective vertex functions are related with the effective couplings constant at scale $k$. We thus define:
\begin{equation}
\Gamma^{(2)}_k(0)=:m^2(k)\,,\quad \Gamma^{(b_2^{(i)})}_k(\{0\})=:(2!)^2\lambda_4(k)\,.\label{rencond}
\end{equation}
In the non-branching sector moreover, strong relations between vertex functions exists due to the recursive definition of melonic diagrams at the perturbation level. Assuming that the structure of these equations can be analytically continued in the non-perturbative level, the following proposition holds:
\begin{proposition}\label{statementClosed}
For the just-renormalizable quartic melonic model in rank-$6$, the melonic non-branching functions $\pi_k^{(b_m^{(i)})}$, for $m=2,3$ satisfy the following relations:
\begin{equation}
\pi^{(b_2^{(i)})}_k(p)=\frac{Z_4\lambda_4}{1+2Z_4\lambda_4 \mathcal{A}_{2k}(p)}\,,
\end{equation}
and:
\begin{equation}
\pi^{(b_3^{(i)})}_k(p)=16 (\pi^{(b_2^{(i)})}_k(p))^3 \mathcal{A}_{3k}(p)
\end{equation}
where:
\begin{equation}
\mathcal{A}_{nk}(p):=\sum_{\textbf{q}\in \mathbb{Z}^d} \theta_\Lambda(\textbf{q}\,^2) \delta_{pq_i} G^n_k(\textbf{q}\,^2)\,,
\end{equation}
for some cut-off function $\theta_\Lambda(\textbf{q}\,^2)$. Note that $\Lambda$ refers to the initial $UV$ scale of the RG flow introduced in section \eqref{sec2.3}.
\end{proposition}
\noindent
To simplify the proof, we use standard intermediate field formalism. This representation maps ordinary Feynman diagrams $\mathcal{G}$ to $\Theta(\mathcal{G})$, explicit construction of the map $\Theta$ being recalled in the Appendix \ref{App1}. We call it $\Theta$-representation. Some steps of the derivations are identical to the one for the rank-5 model without closure constraint investigated in \cite{Lahoche:2016xiq}, and we focus essentially on the specificity of the model, especially in regard to the gauge invariance. Even to come on the proof, we recall some relevant definitions and properties of the $\Theta$-representation.
\begin{definition}
Let $\Theta(\mathcal{G})$ be a Feynman graph in the $\Theta$-representation. The length of any path $C=(u\cdots v)$ linking two-loop vertices $u$ and $v$ is equals to the number of edges building the path.
\end{definition}

\begin{definition}
Let $\mathcal{T}\equiv\Theta(\mathcal{G})$ a tree in the $\Theta$-representation, and $\{\ell_1,\cdots,\ell_n\}\subset \mathcal{T}$ be a subset of $n$ leafs. We call $n$-skeleton the minimal length path linking the $n$-leafs.
\end{definition}

\begin{definition}
The topological $n$-skeleton is obtained from the $n$-skeleton by uncoloring all the edges and replacing all the chains made of $2$-valent loop vertices by a single uncolored solid edge.
\end{definition}

\begin{definition}
We call external leaf, a leaf whose tadpole dotted edge is opened.
\end{definition}
\noindent
Figure \ref{figSkeleton} provides an illustration of the concept of $n$-skeleton and topological $n$-skeleton. As recalled in Appendix \ref{App1}, Theorem \ref{ThMelon}, the melonic diagrams are trees in the $\Theta$-representation, and Lemma \ref{propmelonfaces} can be rephrased as:
\begin{lemma}
Let $\mathcal{G}$ be a melonic diagram with $2N$ external edges. Its $\Theta$-representation corresponds to a tree with $N$ external leafs, linking by a monocolored $N$-skeleton.
\end{lemma}\label{lemmaMelon2}
\noindent
As for lemma \ref{propmelonfaces} the proof follows the derivation of theorem \ref{ThMelon} given in Appendix \ref{App1}.
\begin{figure}
\begin{center}
\includegraphics[scale=1]{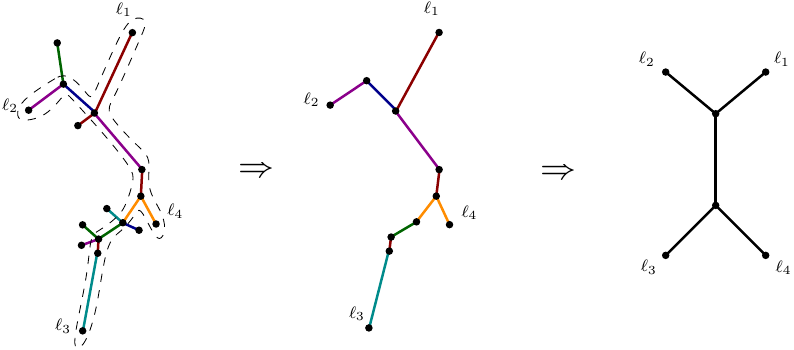}
\end{center}
\caption{On the left: A typical tree with four marked leafs $(\ell_1,\cdots,\ell_4)$. On the middle: the corresponding $4$-skeleton. On the right: The topological $4$-skeleton. }\label{figSkeleton}
\end{figure}

\paragraph{Proof of Proposition \ref{statementClosed}.} Let us consider a typical 1PI Feynman graph contributing to the perturbative expansion to the non-branching melonic vertex function $\pi_k^{(b_2^{(i)})}(p)$. From Lemma \ref{lemmaMelon2}, the $\Theta$-representation of this diagram is a tree with two external leafs linked with a $2$-skeleton $S_i^{(2)}$ of color $i$, with length $[S_i^{(2)}]$. Such a typical tree is pictured in Figure \ref{figtree}. Along the loop vertices building the skeleton, some connected components can be hooked, like $\mathcal{P}_1$ and $\mathcal{P}_2$ on Figure \ref{figtree}. At each corner, all these insertions can be formally resumed from the observation that each insertion is nothing but a piece of the perturbative expansion of the 1PI $2$-point function (self-energy) $\Sigma$. The structure of these local insertions is pictured in Figure \ref{insertions}. Along the corner, the loop is nothing but $C\Sigma^{(1)}C\Sigma^{(3)}C$, where $\Sigma^{(k)}$ denotes the $k$ order term of the self-energy and $C$ the free propagator. Summing over all these terms allowed insertions lead to the formal sum: $C+C\Sigma C+ C\Sigma C \Sigma C + \cdots$, which is the standard Dyson equation for the quantum $2$-point function
\begin{equation}
G=C+C\Sigma G\,.
\end{equation}
As a result, keeping the length $[S_i^{(2)}]=n$ fixed, the sum over all allowed insertions is formally equivalent  to the skeleton (without
insertion),  replacing everywhere the free propagator $C$ by $G$. We denote as $\pi_{nk}^{(b_2^{(i)})}$ this formal sum. In that way, the exact $\pi^{(b_2^{(i)})}_k$ looks like a sum over the length $n$ of the $2$-skeleton:
\begin{equation}
\pi^{(b_4^{(i)})}_k(p)=\sum_{n=1}^\infty\pi_{nk}^{(b_4^{(i)})}(p)\,.
\end{equation}
\begin{figure}
\begin{center}
\includegraphics[scale=1.2]{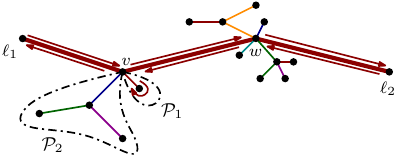}
\end{center}
\caption{A typical tree in the $\Theta$-representation. The red arrows denote the boundaries of the $2$ external monocolored faces. The $2$-skeleton, corresponds to the fat red path $S_i\equiv(\ell_1vw\ell_2)$, of length $[S_i]=3$.}\label{figtree}
\end{figure}
The length $1$ contribution does not depends on $p$, and can be easily computed from standard perturbation theory, and we get:
\begin{equation}
\pi_{1k}^{(b_2^{(i)})}=Z_4 \lambda_4\,.
\end{equation}
In the same way, $\pi_{2k}^{(b_2^{(i)})}$ can be computed from perturbation theory, by replacing at the end of the calculation, $C$ by $G$. We get:
\begin{equation}
\pi_{2k}^{(b_2^{(i)})}=-2Z_4^2\lambda_4^2 \mathcal{A}_{2k}(p)\,,
\end{equation}
For $\pi_{nk}^{(b_2^{(i)})}$, a direct inspection show that it must take the form:
\begin{equation}
\pi_{nk}^{(b_2^{(i)})}= \alpha_n(-1)^{n+1}(Z_4)^{n} \lambda_4^n \left(\mathcal{A}_{2k}(p)\right)^{n-1}\,, \label{decompp}
\end{equation}
where $\alpha_n$ is a purely numerical factor counting the number of configurations, and can be easily computed from perturbation theory, up to the replacement $C\to G$ at the end of the calculation. We have a factor $1/n!$ arising from the exponential, which is cancelled by the permutation of the identical $n$ vertices, say $n!$. Finally, there are two allowed configurations for each of them, leading to $2^n$. The permutation of external edges being given by the factor $4$ in \eqref{decompp}, we have $\alpha_n=2^{n-1}$, and the formal sums of $\pi_{nk}^{(b_4^{(i)})}$'s leads to:
\begin{equation}
\pi^{(b_2^{(i)})}_k(p)=\frac{Z_4\lambda_4}{1+2Z_4\lambda_4 \mathcal{A}_{2k}(p)}\,.
\end{equation}

\begin{figure}
\begin{center}
\includegraphics[scale=1]{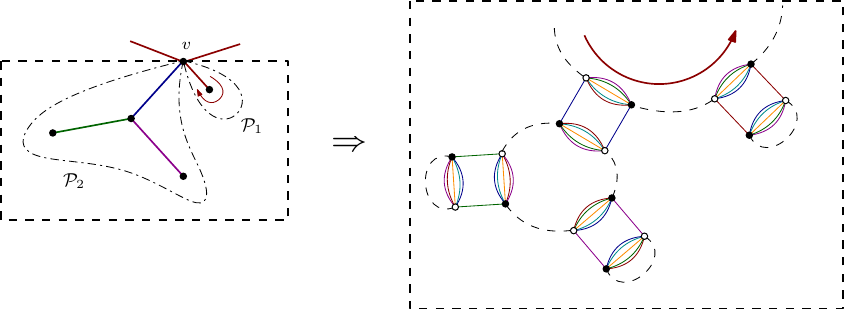}
\end{center}
\caption{Structure of the local insertions along the $2$-skeleton.}\label{insertions}
\end{figure}
Now, let us consider the $6$-point kernel $\pi^{(b_3^{(i)})}_k(p)$. From lemma \ref{lemmaMelon2}, for each Feynman graphs involved in its perturbative expansion, it must exists a $3$-skeleton of color $i$, building a three-armed path between the ending leafs $\ell_1, \ell_2$, $\ell_3$. Moreover, a direct inspection show that all the graphs must have the same topological $3$-skeleton, pictured on Figure \ref{figTopos}. Each arm is nothing but a $2$-skeleton, and the corresponding component on the full diagram can be identified with a contribution to the perturbative expansion of $\pi^{(b_2^{(i)})}_k(p)$, such that each arms can be formally resumed as an effective vertex $\pi^{(b_2^{(i)})}_k(p)$, all of them hooked to the central loop of length three, where the three arms are hooked. As for the other vertices, each corners of this effective vertex can be formally resumed as an effective propagator $G$, and we arrive to the conclusion that:
\begin{equation}
\pi^{(b_3^{(i)})}_k(p)= 8\beta (\pi^{(b_2^{(i)})}_k(p))^3 \mathcal{A}_{3k}(p)\,.
\end{equation}
The remaining numerical factor $\beta$ can be computed from perturbation theory, replacing $C$ by $G$ at the end of the calculation; and it is easy to check that $\beta=2$. For a detailed derivation, see \cite{Lahoche:2019orv}.
\begin{figure}
\begin{center}
\includegraphics[scale=0.8]{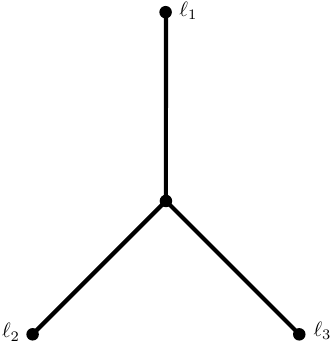}
\end{center}
\caption{The topological $3$-skeleton.}\label{figTopos}
\end{figure}

\begin{flushright}
$\square$
\end{flushright}

This proposition allows to close the infinite hierarchical system of equations obtained by expanding equation \eqref{Wetterich} in terms of vertex functions. If we choose the renormalization conditions for mass and quartic coupling such that:
\begin{equation}
\Gamma^{(2)}_{k=0}(\textbf{p}^2)=m_r^2+\textbf{p}^2+\mathcal{O}(\textbf{p}^2)\,,\quad \pi^{(b_2^{(i)})}_{k=0}(0)=4\lambda_{4}\,, \label{eqrenorm}
\end{equation}
we have the following corollary
\begin{corollary}
Defining the quartic coupling and mass renormalization according to \eqref{eqrenorm}, we have formally:
\begin{equation}
Z_4:=\frac{1}{1-2\lambda_4 \mathcal{A}_{20}(p=0)}\,.
\end{equation}
\end{corollary}
We assume to work at first order in the derivative expansion for $\Gamma_{k}^{(2)}(\textbf{p}^2)$,
\begin{equation}
\Gamma_{k}^{(2)}(\textbf{p}^2)=m^2(k)+Z(k) \textbf{p}^2 \Longleftrightarrow \frac{\partial}{\partial \textbf{p}^2} \Gamma_{k}^{(2)}(\textbf{p}^2)=Z(k)\,, \label{truncation2pts}
\end{equation}
i.e. we truncate around the marginal contributions with respect to the power-counting. We define the anomalous dimension $\eta$ as:
\begin{equation}
\eta(k):= k\frac{d}{dk}\ln(Z(k))\,,
\end{equation}
and we choose to work with the  standard Litim's  regulator \cite{Litim:2000ci}-\cite{Litim:2001dt}:
\begin{equation}
r_k(\textbf{p}^2):=Z(k)(k^2-\textbf{p}^2)\theta(k^2-\textbf{p}^2)\,.
\end{equation}
Note that such a definition makes sense only if the boundary conditions for $k=0$ and $k=\Lambda$ hold. In particular, $\lim_{k\to \Lambda} r_k\sim \Lambda^r$ for $r>0$. Let us consider a fixed point $p$, with anomalous dimension $\eta_p$. In that case $r_{k=\Lambda} \sim \Lambda^{2+\eta_p}$, and we have the physical bounds:
\begin{equation}
\eta_p>-2\,.\label{boundeta}
\end{equation}
We may derive the flow equations in the non-branching melonic sector, assuming the truncation \eqref{truncation2pts} for the $2$-point function. From proposition \ref{statementClosed}, the hierarchy closes in the relevant sector spanned by the dimensionless couplings:
\begin{equation}
\bar{m}^2:=k^{-2}Z^{-1} m^2\,,\qquad \bar{\lambda}_4=:Z^{-2}\bar{\lambda}_4\,. \label{dimensionless}
\end{equation}
Taking the second derivative with respect to $M_\textbf{p}$ and $\bar{M}_\textbf{p}$, we get:
\begin{equation}
k \frac{d}{dk} \Gamma^{(2)}_k(\textbf{p}^2)=-\sum_{\textbf{q}} k\frac{dr_k}{dk}(\textbf{q}^2) \Gamma_k^{(4)}(\textbf{p},\textbf{p},\textbf{q},\textbf{q}) G_k^2(\textbf{q}^2)\,,
\end{equation}
where we introduced the notation
\begin{equation}
\dot{X}:=k \frac{d}{dk}X\,.
\end{equation}
It is suitable to use a graphical representation to materialize the different contributions involved in this equation. In that way we assume that $\Gamma_k^{(2n)}$ decomposes as a sum indexed with bubbles $b$:
\begin{equation}
\Gamma_k^{(2n)}=\sum_b \, \Gamma_{k,b}^{(2n)}\,.
\end{equation}
The bubble $b$ labelling the component $\Gamma_{k,b}^{(2n)}$ being nothing but the boundary graph of the Feynman graphs involved in the perturbative expansion of $\Gamma_{k,b}^{(2n)}$. In that way, the relevant contribution for $\dot{\Gamma}_{k}^{(2)}$ reads as:
\begin{equation}
\dot{\Gamma}_{k}^{(2)}(\textbf{p}^2)=-\sum_{i=1}^d\, \vcenter{\hbox{\includegraphics[scale=1]{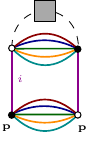}}}\,. \label{eqflow2pts}
\end{equation}
Also, we get for the $4$-point vertex function $\Gamma_k^{(b_2^{(i)})}$:
\begin{equation}
\dot{\Gamma}_k^{(b_2^{(i)})}(\textbf{p}_1,\textbf{p}_2,\textbf{p}_3,\textbf{p}_4)=-\vcenter{\hbox{\includegraphics[scale=0.8]{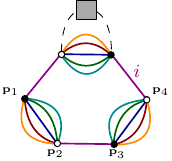}}}+\left( \vcenter{\hbox{\includegraphics[scale=0.8]{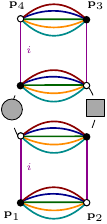}}}+(\textbf{p}_1 \leftrightarrow \textbf{p}_3;\textbf{p}_2 \leftrightarrow \textbf{p}_4 )\right)\,,\label{flowcoupling0}
\end{equation}
where the grey squares and discs materialize respectively the propagators $\dot{r}_k G_k^2$ and the two point function $G_k$:
\begin{equation}
\dot{r}_k(\textbf{p}^2) G_k^2(\textbf{p}^2) \equiv \vcenter{\hbox{\includegraphics[scale=1.2]{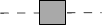}}}\,, \qquad G_k(\textbf{p}^2) \equiv \vcenter{\hbox{\includegraphics[scale=1.2]{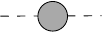}}}\,.
\end{equation}
From the truncation \ref{truncation2pts}, setting $\textbf{p}=0$ in the equation \eqref{eqflow2pts}, we get the flow equation for the mass:
\begin{equation}
\dot{\bar{m}}^2(k)= -\frac{2d \lambda_4(k)Z(k)}{(Z(k)k^2+m^2)^2} \sum_{\textbf{q}\in \mathbb{Z}^{d-1}} \theta(k^2-\textbf{q}^2)\delta\bigg( \sum_{i=1}^{d-1} q_i \bigg) \left[\eta(k)(k^2-\textbf{q}^2)+2k^2 \right]\,.
\end{equation}
Then we focus on the flow equations in the UV sector ($\Lambda \gg k \gg 1$) then it is suitable to replace the sums with integrals. The integration domain corresponds to the intersection between the $(d-1)$-ball of radius $k$ and the hyperplane of equation $\sum_{i=1}^{d-1} q_i=0$. Now let us define:
\begin{equation}
S_n(p^2):= \int d^{d-1} q\,(p^2+\textbf{q}^2)^n \theta(k^2-p^2-\textbf{q}^2)\delta\bigg(p+ \sum_{i=1}^{d-1} q_i\bigg)\,.
\end{equation}
We can straightforwardly compute this integral as the intersection between $p+\sum_{i=1}^{d-1} q_i=0$ and the $(d-1)$-ball of radius $\sqrt{k^2-p^2}$. We get after a calculation:
\begin{equation}
S_n(p^2)=\frac{\Omega_{d-2}}{\sqrt{d-1}} \left(k^2-\frac{d}{d-1}p^2 \right)^{\frac{d-2}{2}}\sum_{k=1}^n \binom{n}{k} \left(\frac{d}{d-1}p^2 \right)^{n-k}\frac{1}{2k+d-2}\left(k^2-\frac{d}{d-1}p^2 \right)^k\,,
\end{equation}
where
$\Omega_{d}:=\frac{\pi^{d/2}}{\Gamma(d/2+1)}$ denotes the volume of the unit $d$-ball. In particular:
\begin{equation}
S_0(p^2)=\frac{\Omega_{d-2}}{\sqrt{d-1}}\left(k^2-\frac{d}{d-1}p^2 \right)^{\frac{d-2}{2}}\,,
\end{equation}
and
\begin{equation}
S_1(p^2)=\frac{\Omega_{d-2}}{\sqrt{d-1}}\left(k^2-\frac{d}{d-1}p^2 \right)^{\frac{d-2}{2}} \left[\frac{d}{d-1} p^2+\frac{d-2}{d}\left(k^2-\frac{d}{d-1}p^2 \right) \right]\,.
\end{equation}
Finally:
\begin{equation}
\dot{{m}}^2(k)= -\frac{\Omega_{d-2}k^d}{\sqrt{d-1}}\frac{4d \lambda_4(k)Z(k)}{(Z(k)k^2+m^2)^2} \left(1+\frac{\eta(k)}{d}\right)\,.
\end{equation}
Setting $d=6$, and using dimensionless couplings $\bar{m}^2$ and $\bar{\lambda}_4$ defining in \eqref{dimensionless} rather than $m^2$ and $\lambda_4$, we get ($\Omega_4=\pi^2/2$):
\begin{equation}
\boxed{
\beta_{m^2}=-(2+\eta)\bar{m}^2-\frac{\pi^2}{\sqrt{5}}\frac{12 \bar{\lambda}_4(k)}{(1+\bar{m}^2)^2} \left(1+\frac{\eta(k)}{6}\right)\,,}\label{betam}
\end{equation}
where we introduced the notation $\beta_X:=k \frac{d}{dk} \bar{X}$. In the same way, setting: $\textbf{p}_1=\textbf{p}_2=\textbf{p}_3=\textbf{p}_4=0$, from equation \eqref{rencond} and proposition \ref{statementClosed}, we get straightforwardly:
\begin{equation}
4\dot{\lambda}_4=-24 \frac{\Omega_{d-2}k^d}{\sqrt{d-1}}\frac{\pi_k^{(b_3^{(i)})}(0)Z(k)}{(Z(k)k^2+m^2)^2} \left(1+\frac{\eta(k)}{d}\right)+8\lambda_4^2\frac{\Omega_{d-2}k^d}{\sqrt{d-1}} \frac{1+\frac{\eta}{d}}{(Zk^2+m^2)^3}\,,
\end{equation}
which is expressed in terms of the dimensionless parameters as:
\begin{equation}
\beta_{\lambda_4}=-2\eta\bar{\lambda}_4-\frac{6\Omega_{d-2}}{\sqrt{d-1}}\frac{\bar{\pi}_k^{(b_3^{(i)})}(0)}{(1+\bar{m}^2)^2} \left(1+\frac{\eta}{d}\right)+8\bar{\lambda}_4^2 \frac{\Omega_{d-2}}{\sqrt{d-1}} \frac{1+\frac{\eta}{d}}{(1+\bar{m}^2)^3}\,.
\end{equation}
To compute $\bar{\pi}_k^{(b_3^{(i)})}$ we need to estimate the sum $\mathcal{A}_{3k}(0)$. This sum is not constrained in the windows of momenta allowed by the regulator $\dot{r}_k$. In our previous investigation we showed that the integrals which are superficially convergent can be computed using the same approximations as used to compute flow integrals, and checked consistency with WT identities \cite{Lahoche:2018oeo}.  Here, we have no Ward identities to support such an approximation, but we assume its validity again. Explicitly the sum reads:
\begin{equation}
\mathcal{A}_{3k}(0)=\sum_{q\in \mathbb{Z}^5}\delta\left(\sum_i q_i\right) \left[\frac{\theta_k(q^2)}{(Zk^2+m^2)^3}+\frac{\theta_\Lambda(q^2)-\theta_k(q^2)}{(Zp^2+m^2)^3}\right]\,.
\end{equation}
We define the dimensionless quantity as $\bar{\mathcal{A}}_{3k}$ as:
\begin{equation}
\bar{\mathcal{A}}_{3k}(0)=Z^3 k^{-2} {\mathcal{A}}_{3k}(0)\,,
\end{equation}
and we get:
\begin{equation}
\bar{\mathcal{A}}_{3k}(0)=\frac{\pi^2}{2\sqrt{5}}\frac{1}{1+\bar{m}^2} \left[ \frac{1}{(1+\bar{m}^2)^2}+\left(1+\frac{1}{1+\bar{m}^2}\right) \right]\,.
\end{equation}
Then  
\beq
\bar{\pi}_k^{(b_3^{(i)})}(0)=16(\lambda_4)^3Z^3k^{-2}\bar{\mathcal{A}}_{3k}(0).
\eeq
Therefore $\beta_{\lambda_4}$ reads as:
\begin{equation}
\boxed{
\beta_{\lambda_4}=-2\eta\bar{\lambda}_4+\frac{4\pi^2\bar{\lambda}_4^2}{\sqrt{5}}\frac{1+\frac{\eta}{6}}{(1+\bar{m}^2)^3}\left[1-\frac{12\pi^2\bar{\lambda}_4}{\sqrt{5}}\left( \frac{1}{(1+\bar{m}^2)^2}+\left(1+\frac{1}{1+\bar{m}^2}\right) \right)\right]\,.}\label{betalambda}
\end{equation}
The remaining investigation is the computation of the anomalous dimension $\eta$. An explict expression can be derived from the flow equation \eqref{eqflow2pts}. Because of the definition \eqref{truncation2pts}, taking derivative with respect to $p_1^2$ and setting $\textbf{p}=0$, we get:
\begin{equation}
\dot{Z}=-\frac{4Z\Omega_4 k^6}{\sqrt{5}}\frac{1+\frac{\eta}{6}}{(Zk^2+m^2)^2}\frac{d\pi_k^{(b_2^{(1)})}}{dp_1^2}\bigg\vert_{p_1=0}-\frac{2\lambda_4}{(Zk^2+m^2)^2}\left[\eta(k^2S_0^\prime(0)-S_1^\prime(0))+2k^2S_0^\prime(0)\right]\,.\label{derivvertex}
\end{equation}
The derivatives of $S_1$ and $S_0$ are:
\begin{equation}
S_0^\prime(0)=-\frac{6\pi^2}{5\sqrt{5}}k^2\,, \quad S_1^\prime(0)=-\frac{3\pi^2}{5\sqrt{5}}k^4\,,
\end{equation}
and:
\begin{equation}
\eta(k^2S_0^\prime(0)-S_1^\prime(0))+2k^2S_0^\prime(0)=-k^4\eta \frac{3\pi^2}{5\sqrt{5}}-k^4\frac{12\pi^2}{5\sqrt{5}}\,.
\end{equation}
The term involving a derivative of the vertex in \eqref{derivvertex} comes from the EVE, which takes into account the full momentum dependence of the effective vertices. From proposition \ref{statementClosed}, we get:
\begin{equation}
\frac{d\pi_k^{(b_2^{(1)})}}{dp_1^2}\bigg\vert_{p_1=0}=-2\lambda_4^2(k) \frac{d}{d p_1^2}\mathcal{A}_{2k}(p_1=0)\,.
\end{equation}
We then compute this derivative to obtain a tractable expression for the anomalous dimension. Because the derivative is a superficially convergent quantity, we use of the same approximation as for $\mathcal{A}_{3k}$, using ansatz \eqref{truncation2pts}. Once again, for the models without closure constraint, this approximation is controlled by the WT identities \cite{Lahoche:2018oeo}. We expect that the same approximation makes sense here, in regard to the convergence of the integral. We get explicitly:
\begin{equation}
\frac{d}{d p_1^2}\mathcal{A}_{2k}(p_1=0)=\frac{1}{\sqrt{5}}\frac{\pi^2}{k^2Z^2}\frac{1}{1+\bar{m}^2}\left(1+\frac{1}{1+\bar{m}^2}\right)\,.
\end{equation}
Solving \eqref{derivvertex} for $\eta$, we thus obtain:
\begin{equation}\label{eta}
\boxed{
\eta=\frac{24\pi^2\bar{\lambda}_4}{5\sqrt{5}} \frac{(1+\bar{m}^2)^2-\frac{\pi^2\sqrt{5} \bar{\lambda}_4}{6}(2+\bar{m}^2)}{(1+\bar{m}^2)^4-\frac{6\pi^2 \bar{\lambda}_4}{5\sqrt{5}}(1+\bar{m}^2)^2+\frac{\pi^4\bar{\lambda}_4^2}{15}(2+\bar{m}^2)}\,.}
\end{equation}

\subsection{Fixed point solutions}
The autonomous system given by equations \eqref{betam}, \eqref{betalambda} and \eqref{eta} can be solved numerically. In addition to the Gaussian fixed point (GFP) We find three real fixed points:
\begin{equation}
FP1:=\{\bar{m}^2\approx 0.30, \bar{\lambda}_4\approx 0.34 \}\,,
\end{equation}
\begin{equation}
FP2=\{\bar{m}^2\approx--0.16, \bar{\lambda}_4\approx -0.25 \}\,,
\end{equation}
\begin{equation}
FP3=\{\bar{m}^2\approx 0.32, \bar{\lambda}_4\approx -0.02\}\,.
\end{equation}
However, none of these fixed points are physically relevant. Indeed, the anomalous dimensions of the two first ones are respectively $\eta_{1}\approx -5.4$ and $\eta_2\approx -5.8$, under the physical bound \eqref{boundeta}. The third fixed point FP3 has anomalous dimension $\eta_3 \approx -0.25$, but the coupling constant is negative, and therefore has the wrong sign with respect to the stability requirement. 
We therefore arrive at different conclusions from \cite{Benedetti:2015yaa}. The authors of the paper argued for the existence of a non-Gaussian fixed point of the Wilson-Fisher type in the UV regime, reminiscent of a second order phase transition. The EVE on the other hand shows that this fixed point must be an artifact of the truncation. In the next section we will confirm this conclusion by considering larger truncations than the one of these authors, showing clearly the instability of the fixed point solution. 
\medskip

Despite the disappearance of the non-perturbative fixed point, the theory is asymptotically free, and around the Gaussian fixed point we get:
\begin{equation}
\beta_{\lambda_4}=-\frac{28\pi^2\bar{\lambda}_4^2}{5\sqrt{5}}\,.
\end{equation}
Note that there is moreover a strong improvement using EVE method with respect to the results given in \cite{Benedetti:2015yaa}, which used ordinary vertex expansion, without taking into account the momentum dependence of the vertex. From this method, the anomalous dimension exhibits a singularity above the singularity $\bar{m}^2=-1$, while it is moved below using EVE. Moreover, the fixed point FP2 was localised above $\bar{m}^2=-1$ but below the anomalous dimension singularity. Using EVE, we show that this fixed point is pushed below the singularity $\bar{m}^2=-1$, and disconnected from the Gaussian region.

\subsection{Convergence of the vertex expansion}
In this section we investigate the convergence of the ordinary vertex expansion in the subregion of the full theory space spanned with non-branching melons. We denote as $\lambda_{2n}$ the coupling associated to the non-branched melon with $n$ black nodes ($\lambda_2\equiv m^2$). In \cite{Carrozza:2016tih} the authors showed that the corresponding $\beta$-functions can be recursively defined, and:
\begin{equation}
\beta_{\lambda_{2n}}=(2n-4-n\eta)\bar{\lambda}_{2n}+\frac{\pi^2}{\sqrt{5}}\left(1+\frac{\eta}{6}\right)\sum_{k=1}^n (-1)^k \frac{1}{(1+\bar{m}^2)^{k+1}}\sum_{\{n_{2q}\}\in\mathcal{D}_{k,p}}\frac{k!}{\prod_{q\geq 2} n_{2q}!}\prod_{q\geq 2} (q\bar{\lambda}_{2q})^{n_{2q}}\,.
\end{equation}
In that sum, the set $\mathcal{D}_{k,p}$ denotes the set of integer solutions for the two constraints:
\begin{equation}
\sum_{q>1} n_{2q}=k\,,\qquad \sum_{q>1} q n_{2q}=p+k\,.
\end{equation}
We provide the results up to order $\phi^{12}$ in the truncation. The resulting equations can be investigated numerically, and we obtain many fixed point solutions. First of all, up to order $4$, we get two non-reliable fixed point solutions reminiscent of $FP2$ and $FP3$, having a large and negative anomalous dimension $\eta \lesssim -7$. A fixed point having the same characteristics as the fixed point discovered in \cite{Benedetti:2015yaa}, say FP0 is also recovered to all orders, excepts notably for $\phi^8$ truncation, where no reliable fixed point is found. Note that such instability was noticed in \cite{Carrozza:2016tih}-\cite{Carrozza:2017vkz}. The results about this fixed point are summarized in Table \ref{tab1}, where we indicate the values of the anomalous dimension, of the relevant coupling and the relevant direction toward IR scales.
Note moreover that, except for the discontinuity for octic truncation, the remaining critical exponents do not vary too, two of them becoming complex  from order $\phi^{10}$ on. Finally, another fixed point appears up to order $\phi^{10}$, which has two relevant directions toward IR scales. The corresponding anomalous dimension is of order $1$. However, due to its disappearance for lower truncation, it is difficult to see it other than as an anomaly linked to the truncation.

\begin{center}
\begin{tabular}{|C{4cm}||C{1cm}|C{1.5cm}|C{1.5cm}|C{1.5cm}|C{1.5cm}|}
\hline Truncation order & 4 & 6 & 8 & 10 & 12 \\
\hline $\eta$ & 0.67 & 1.10 & $\emptyset$ & 0.60 & 0.60 \\
\hline
$\theta_1$ & 0.90 & 1.01 & $\emptyset$ & 1.95& 1.98\\
\hline
$u_2$ & -0.60 & -0.82 & $\emptyset$ & -0.55 & -0.55 \\
\hline $u_4$ & 0.004 & 0.0012 & $\emptyset$ & 0.004 & 0.005\\
\hline
\end{tabular}
\captionof{table}{The characteristics of the fixed point FP0 in the melonic non-branching sector up to order $\phi^{12}$ in the truncation. $\theta_1$ denotes the critical exponent corresponding to the single relevant direction toward IR scales.}
\label{tab1}
\end{center}

 Note that, it is possible to study the dependence of the regulator in  Table \eqref{tab1}, as a continuation of our recent results \cite{Lahoche:2019ocf}-\cite{Lahoche:2020pjo}. This requires an analysis of the effect of closure constraint, but also the definition and   the choice of optimal regulators adopted for such a model. This can be the purpose of  a future investigation. 

\section{Discussion and conclusion} \label{sec4}
In this paper, we investigated a non-perturbative solution of the exact RG equation for a rank-6 TGFT with a closure constraint. We focused on the symmetric phase, and along a subsector of the full theory space, spanned with non-branching melons. Using the non-trivial relations between effective vertices which holds along the RG trajectories and corresponding to the just-renormalizable model, we may be able to close the hierarchical RG flow equations around quartic interactions, keeping the full momenta dependence of effective vertices through an improved version of the EVE. This non-trivial improvement for standard vertex expansion has two relevant consequences: (i) The announced fixed-point solution in \cite{Benedetti:2015yaa} disappears. (ii) The anomalous dimension is non-singular in the region $\bar{m}^2>-1$.
\medskip

The particular aspect of this work comes from a novelty of the method used to derive the flow equation. Indeed, local truncations, by ignoring momentum dependence of effective vertices, we exhibits a singularity line \cite{Benedetti:2015yaa} above the singularity $\bar{m}^2=-1$. In that sense, EVE provides a maximal extension of the explored region, bounded from below with the singularity $\bar{m}^2=-1$, which may be physically interpreted as a consequence of the zero-field expansion (symmetric phase). Moreover, while standard vertex expansion provides similar results at low orders, higher-order contributions deviate significantly or are pathological (as was the case for the octic truncation). As we had pointed out in the conclusion of our previous analysis \cite{Lahoche:2018ggd}, no WT identity violation is expected, as it was the case for models without closure constraint, and this becomes an important indication in favor of geometrical inputs arising from canonical quantum gravity for discrete quantum gravity models like GFTs. \\

\section*{Acknowledgments}
 The authors are grateful to the referee for his useful comments and remarks that allowed to improve the paper.

\appendix
\begin{center}
\Large{\textbf{Appendix}}
\end{center}
\section{Leading order, just-renormalizability and canonical dimension}\label{App1}

In this section we introduce the standard intermediate field formalism and investigate some basics properties of the leading order (melonic) sector. The intermediate field representation is introduced as a “trick" coming from basics properties of the Gaussian integration, which allows breaking a quartic interaction as a three-body interaction for two fields; the archetypal example of such a strategy being provided by the $W$ and $Z$ bosons of the weak nuclear interactions. However, to simplify the presentation, we introduce the intermediate field decomposition as a one-to-one correspondence between Feynman graphs, following
\cite{Lahoche:2015yya}-\cite{Lahoche:2015zya} and references therein. Let us consider a vacuum graph $\mathcal{G}$. For the quartic model that we consider, such a graph looks like a connected set of quartic melonic interactions, whose black and white nodes are linked with some closed dotted edges. The intermediate field correspondence work is as follow: To each closed-loop build as a closed cycle of dotted edges, we associate a vertex, whose valence (i.e. the number of corners) equals the length of the loop (i.e. the number of dotted edges building it). We call them \textit{loop vertices}. The original vertices are of six different types, each of them being labeled with a color $i\in\{1,\cdots,6\}$. Each of these vertices is mapped as a colored link, of the same color as the label $i$ for the original vertex. The procedure may be materialized with a map $\Theta:\mathcal{G}\to \Theta(\mathcal{G})$, where the figure \ref{figMapIntermediate} illustrates the construction on some explicit examples. We shall use the $\Theta$-representation or intermediate field representation to sketch the proof of the following statement:
\begin{theorem}\label{ThMelon}
The 1PI leading order vacuum graphs are trees in the intermediate field representation. In the original representation, these trees are called melonic diagrams. Moreover, the perturbative expansion of the model is power counting just renormalizable.
\end{theorem}
\paragraph{Proof.} As a preliminary observation, let us note that from construction, the colored propagators $c_i$ associated with colored edges of type $i$ in the $\Theta$-representation does not depend on the color $i$, and looks like a matrix sharing a pair of indices $c_i\equiv \{(c_i)_{p_ip_i^\prime}\}$. However, due to the closure constraint, $(c_i)_{p_ip_i^\prime}$ must be diagonal $(c_i)_{p_ip_i^\prime}\equiv a_i(p_i) \delta_{p_ip_i^\prime}$. This can be checked as follows. Figure \ref{FigDag} show the structure of a node in the $\Theta$-representation. Because of the closure constraint, the momenta $\textbf{p}$ and $\textbf{p}^\prime$ have to satisfy $\sum_i p_i=0$ and $\sum_i p_i^\prime=0$. Moreover, because the $d-1$ colored edges between black and white nodes are Kronecker deltas, they impose $p_j=p_j^\prime$, $\forall\,j\neq i$; implying $p_i=p_i^\prime$. Therefore we shall prove the Theorem \ref{ThMelon} by recurrence. Let us consider the following Lemma:
\begin{lemma}
For any Feynman graph $\mathcal{G}$ in the original representation, with $L$ internal (dotted) edges, $F$ closed faces and $V$ vertices, the divergent degree $\omega(\mathcal{G})$ is given by:
\begin{equation}
\omega(\mathcal{G})=-2L(\mathcal{G})+(F(\mathcal{G})-R(\mathcal{G}))\,,
\end{equation}
where $R(\mathcal{G})$ denotes the rank of the incidence matrix $\epsilon_{ef}$, whose entries are equal to 1 if $e\in \partial f$ or $0$ otherwise\footnote{In principle, for $e\in \partial f$, $\epsilon_{ef}=\pm 1$, depending on the relative orientation of $e$ and $f$. For the bipartite model that we consider however, relative orientation can be fixed unambiguously, such that $\epsilon_{ef}=\pm 1$.}.\end{lemma}
The proof can be found in \cite{Lahoche:2015ola} using multi-scale decomposition.
\medskip
\begin{figure}
\begin{center}
\includegraphics[scale=0.8]{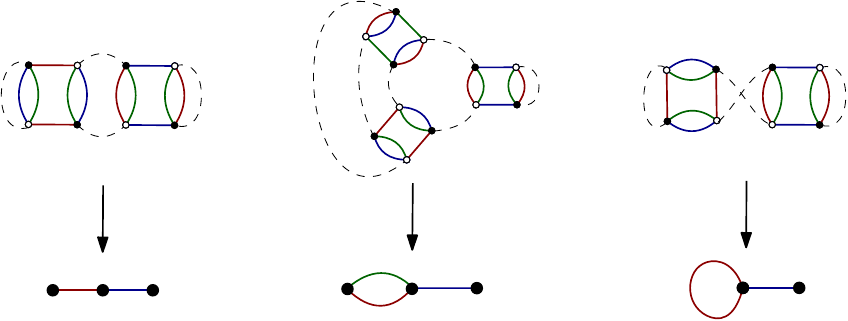}
\end{center}
\caption{Illustration of the mapping from original to intermediate field formalism for vacuum graphs. Original vertices become colored links and loops become vertices.}\label{figMapIntermediate}
\end{figure}

\begin{figure}
\begin{center}
\includegraphics[scale=1.5]{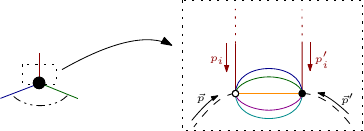}
\end{center}
\caption{Structure of the contact point between a colored (red) edge and a loop vertex.}\label{FigDag}
\end{figure}
Let us prove the first assertion of the theorem from a recursion on the number $\ell$ of edges in the $\Theta$-representation. For $\ell=1$, there are two diagrams:
\begin{equation}
\vcenter{\hbox{\includegraphics[scale=1]{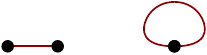}}}
\end{equation}
For the first one, $F=11$, $L=2$ and $R=2$, thus $\omega=-4+(11-2)=5$. For the second one $F=7$, $L=2$ and $R=2$, thus $\omega=1$; and the theorem holds. Now, let us consider a tree $\mathcal{T}_n$ in the $\Theta$-representation, having $\ell=n$. It has the following structure:
\begin{equation}
\vcenter{\hbox{\includegraphics[scale=1]{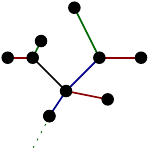}}}\,.
\end{equation}
From such a tree, four elementary moves are allowing to construct a $\Theta$-graph with $\ell=n+1$ from this one, all illustrated on Figure \ref{figrecur}. The move $a$ creates two corners and at least one face, and $\delta R=0$, leading to $\delta \omega=-3$. For $b$, we create a single face, increasing in the same time the rank by one, and $\omega=-4$. Finally, for $c$ and $d$, which preserves the tree structure, we increase the number of faces by $5$, the rank by $1$, and $\delta \omega=0$. The move which preserves the tree structure and preserves also the divergent degree, and we proved that trees in the $\Theta$-representation are the most divergent diagrams. We consider the following definition:
\begin{definition}
In the $\Theta^{-1}$-(original) representation, these trees is said to be melonic diagrams.
\end{definition}
\begin{figure}
\begin{center}
\includegraphics[scale=1]{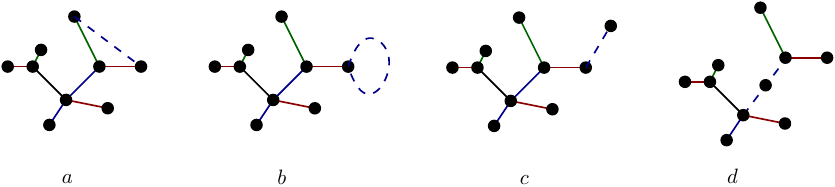}
\end{center}
\caption{The four moves allowing to construct a $n+1$ graph from $\mathcal{T}_n$.}\label{figrecur}
\end{figure}
Let us prove the second part of the theorem i.e. the power-counting renormalizability. To start, let us note that the leading order non-vacuum graph can be obtained from leading order vacuum graphs by deleting some internal (dotted) edges optimally. This method allows obtaining the structure of 1PI leading order graphs. Let us consider a 1PI vacuum diagram $\mathcal{G}=\Theta^{-1}(\mathcal{T})$. From the structure of the graph, it is clear that the deleted dotted edges have to be a tadpole, to keep the 1PI structure. Deleting one such tadpole discard one dotted edge, $d=6$ faces, and the rank decreases to $1$. As a consequence, the power counting decreases by $\delta \omega=-3$, and the leading order 1PI $2$-point function has divergent degree $\omega=2$. Deleting another edge, we construct a $4$-point diagram. For the same reason, the resulting diagram will be 1PI if and only if the deleted edge is a tadpole. Moreover, the deletion will be optimal if the deleted tadpole is along the boundary of an opened face. In that way we lost one internal edge, $d-1=5$ faces and the rank decreases by one. The variation of the divergent degree is therefore: $\delta \omega= 2-(5-1)=-2$, and the leading order $1PI$ $4$-point amplitudes have divergent degree $\omega=0$. Recursively, we prove the lemma \ref{propmelonfaces}. To check the power-counting just renormalizability, let us note that contracting the dotted edges (following the procedure of definition \ref{defsum}) along a spanning tree (in the $\Theta^{-1}$ representation) does not change the number of faces and the rank $R$, but increases by $2(V-1)$ the divergent degree\footnote{For a Feynman graph having $V$ vertices, the number of edges building a spanning tree is $V-1$. Note that, morally, such a contraction has not to change the behavior of the corresponding Feynman amplitude concerning some UV cut-off. Indeed, from definition \ref{defsum}, each contraction builds the sum of the two adjacent vertices, whose canonical dimension (see below) exactly compensate for the change of the power counting. This means that the "divergent degree does not change".}. The resulting graph has divergent degree $\tilde{\omega}=\omega+2(V-1)$; but is build as a single big vertex contracted with $L-V+1$ edges. These edges can then be contracted optimally. From lemma \ref{propmelonfaces}, we know that the contraction will be optimal if the contracted dotted edge is along the boundary of the mono-colored external faces running through the diagram. We can consider successive $(d-1)$-dipole contractions as follows. We recall that a $k$-dipole is made with two black and white nodes (in the original representation), wished together with one dotted edge and $k$ colored edges. In the intermediate field representation, and for a vacuum diagram, we can proceed both on the tree and $d-1$ dipole contraction as follow. Let us consider a leaf hooked to a loop vertex $b$ with $p$ external edges hooked to him. Contracting the leaf, we discard $(d-1)$-faces (or $d$ if it is the first one that we delete), $1$ dotted edge, and the rank $R$ decreases by $1$. We may assume that only leafs are hooked to $b$, except for one colored edge. Using the same procedure for all these leafs, we get an effective loop of length $p$, and we can select $p-1$ edges along the spanning tree. Deleting them, we get a tadpole, and we can repeat the same procedure, leading to $F-R=(d-2)(L-V+1)+1$, the $+1$ coming from the first contraction, which deletes $d$ faces. For a non-vacuum graph with $2N$ external edges, creating them cost $d-1$ faces per deleted tadpole, except for the first one, which cost $d$ faces, therefore: $F-R=(d-2)(L-V+1)$. Using the topological relation $2L=4V-N_{ext}$ holding for a quartic model, where $N_{ext}$ denotes the number of external edges, we get (see also equation \eqref{powercountingTrue}):
\begin{equation}
\omega=-4V+N_{ext}+2(4V-N_{ext}/2-2V+2)=4-N_{ext}\,.\label{melocountinplus}
\end{equation}
Which prove power counting just-renormalizability of the quatic melonic model.
\begin{flushright}
$\square$
\end{flushright}
To conclude this appendix, let us briefly discuss the scaling dimension problem. In the standard field theory context, a canonical notion of dimension is inherited from the background space. For a background independent field theory as GFT however, there is no such a canonical notion of dimension a priori. There are at this stage two ways to introduce a dimension in that context. The first way is to fix the dimension from physical considerations, from the contacts between GFT and the theory having connections with space-time like LQG. Another, the more abstract way is to fix the dimension from the power-counting itself. From the previous power counting, we know that leading order $4$-point diagrams have to scale as $\ln(\Lambda)$ with respect to some UV cut off $\Lambda$, and it is tempting to attribute a dimension $0$ for the quartic melonic coupling. In general, radiative corrections behaves like $\Lambda^n$; and we define the canonical dimension as the optimal $n$, i.e., following the behavior of the leading order quantum corrections.
\begin{definition}
Let $\mathcal{B}$ be a bubble having $2n$ black and white nodes. Let $\mathbb{A}(\mathcal{B}):=\{\mathcal{A}_\mathcal{B}\}$ the set of $2$-point graphs obtained from $\mathcal{B}$. The canonical dimension $d_\mathcal{B}$ of $\mathcal{B}$ is then defined as:
\begin{equation}
d_\mathcal{B}:= 2-\max_{\mathcal{A}_\mathcal{B}\in\mathbb{A}(\mathcal{B})} \, \omega(\mathcal{A}_\mathcal{B})\,.
\end{equation}
\end{definition}



\begin{thebibliography}{99}
 \small
\bibitem{Jegerlehner:2021vqz}
F.~Jegerlehner,
``The Standard Model of Particle Physics as a Conspiracy Theory and the Possible Role of the Higgs Boson in the Evolution of the Early Universe,''
Acta Phys. Polon. B \textbf{52} (2021) no.6-7, 575-605
doi:10.5506/APhysPolB.52.575
[arXiv:2106.00862 [hep-ph]].
%
%
 
\bibitem{Hawking:1979ig} 
  S.~W.~Hawking and W.~Israel,
  ``General Relativity : An Einstein Centenary Survey,''



\bibitem{Donoghue:1994dn} 
  J.~F.~Donoghue,
  ``General relativity as an effective field theory: The leading quantum corrections,''
  Phys.\ Rev.\ D {\bf 50}, 3874 (1994)
  doi:10.1103/PhysRevD.50.3874
  [gr-qc/9405057].
 
\bibitem{Witten:2002wb}
E.~Witten,
``Comments on string theory,''
[arXiv:hep-th/0212247 [hep-th]].
 
\bibitem{Rovelli:1997yv} 
  C.~Rovelli,
  ``Loop quantum gravity,''
  Living Rev.\ Rel.\  {\bf 1}, 1 (1998)
  doi:10.12942/lrr-1998-1
  [gr-qc/9710008].



\bibitem{Rovelli:1998gg} 
  C.~Rovelli and P.~Upadhya,
  ``Loop quantum gravity and quanta of space: A Primer,''
  gr-qc/9806079.
 
\bibitem{Ambjorn:1992eh} 
  J.~Ambjorn, Z.~Burda, J.~Jurkiewicz and C.~F.~Kristjansen,
  ``Quantum gravity represented as dynamical triangulations,''
  Acta Phys.\ Polon.\ B {\bf 23}, 991 (1992).

\bibitem{Ambjorn:1995jt} 
  J.~Ambjorn,
  ``Quantum gravity represented as dynamical triangulations,''
  Class.\ Quant.\ Grav.\  {\bf 12}, 2079 (1995).
  doi:10.1088/0264-9381/12/9/002


\bibitem{Ambjorn:2013tki} 
  J.~Ambjørn, A.~Görlich, J.~Jurkiewicz and R.~Loll,
  ``Quantum Gravity via Causal Dynamical Triangulations,''
  doi:10.1007/978-3-642-41992-8-34
  arXiv:1302.2173 [hep-th].


\bibitem{Oriti:2009nd} 
  D.~Oriti,
  ``Group field theory and simplicial quantum gravity,''
  Class.\ Quant.\ Grav.\  {\bf 27}, 145017 (2010)
  doi:10.1088/0264-9381/27/14/145017
  [arXiv:0902.3903 [gr-qc]].
 
\bibitem{Rovelli:2011eq}
C.~Rovelli,
``Zakopane lectures on loop gravity,''
PoS \textbf{QGQGS2011} (2011), 003
doi:10.22323/1.140.0003
[arXiv:1102.3660 [gr-qc]].
 

\bibitem{Markopoulou:2006qh}
F.~Markopoulou,
``Towards Gravity from the Quantum,''
[arXiv:hep-th/0604120 [hep-th]].
%

\bibitem{Steinhauer:2015saa}
J.~Steinhauer,
``Observation of quantum Hawking radiation and its entanglement in an analogue black hole,''
Nature Phys. \textbf{12} (2016), 959
doi:10.1038/nphys3863
[arXiv:1510.00621 [gr-qc]].
%



\bibitem{Rivasseau:2014ima} 
  V.~Rivasseau,
  ``The Tensor Theory Space,''
  Fortsch.\ Phys.\  {\bf 62}, 835 (2014)
  doi:10.1002/prop.201400057
  [arXiv:1407.0284 [hep-th]].


\bibitem{Gurau:2009tw} 
  R.~Gurau,
  ``Colored Group Field Theory,''
  Commun.\ Math.\ Phys.\  {\bf 304}, 69 (2011)
  doi:10.1007/s00220-011-1226-9
  [arXiv:0907.2582 [hep-th]].
 
\bibitem{Gurau:2011xp} 
  R.~Gurau and J.~P.~Ryan,
  ``Colored Tensor Models - a review,''
  SIGMA {\bf 8}, 020 (2012)
  doi:10.3842/SIGMA.2012.020
  [arXiv:1109.4812 [hep-th]].







\bibitem{Gurau:2011xq} 
  R.~Gurau,
``The complete 1/N expansion of colored tensor models in
arbitrary dimension,''
  Annales Henri Poincare {\bf 13}, 399 (2012)
  doi:10.1007/s00023-011-0118-z
  [arXiv:1102.5759 [gr-qc]].



\bibitem{Gurau:2010ba} 
  R.~Gurau,
  ``The 1/N expansion of colored tensor models,''
  Annales Henri Poincare {\bf 12}, 829 (2011)
  doi:10.1007/s00023-011-0101-8
  [arXiv:1011.2726 [gr-qc]].


\bibitem{Gurau:2013pca} 
  R.~Gurau,
  ``The 1/N Expansion of Tensor Models Beyond Perturbation Theory,''
  Commun.\ Math.\ Phys.\  {\bf 330}, 973 (2014)
  doi:10.1007/s00220-014-1907-2
  [arXiv:1304.2666 [math-ph]].
 
 
\bibitem{Gurau:2013cbh} 
  R.~Gurau and J.~P.~Ryan,
  ``Melons are branched polymers,''
  Annales Henri Poincare {\bf 15}, no. 11, 2085 (2014)
  doi:10.1007/s00023-013-0291-3
  [arXiv:1302.4386 [math-ph]].
  %
 
 
  


\bibitem{Bonzom:2012hw} 
  V.~Bonzom, R.~Gurau and V.~Rivasseau,
  ``Random tensor models in the large N limit: Uncoloring the colored tensor models,''
  Phys.\ Rev.\ D {\bf 85}, 084037 (2012)
  doi:10.1103/PhysRevD.85.084037
  [arXiv:1202.3637 [hep-th]].

\bibitem{Bonzom:2016dwy} 
  V.~Bonzom,
  ``Large $N$ Limits in Tensor Models: Towards More Universality Classes of Colored Triangulations in Dimension $d\geq 2$,''
  SIGMA {\bf 12}, 073 (2016)
  doi:10.3842/SIGMA.2016.073
  [arXiv:1603.03570 [math-ph]].
  %
  


  
 
 
\bibitem{Benedetti:2017fmp} 
  D.~Benedetti, S.~Carrozza, R.~Gurau and A.~Sfondrini,
  ``Tensorial Gross-Neveu models,''
  JHEP {\bf 1801}, 003 (2018)
  doi:10.1007/JHEP01(2018)003
  [arXiv:1710.10253 [hep-th]].
 


\bibitem{Klebanov:2016xxf} 
  I.~R.~Klebanov and G.~Tarnopolsky,
  ``Uncolored random tensors, melon diagrams, and the Sachdev-Ye-Kitaev models,''
  Phys.\ Rev.\ D {\bf 95}, no. 4, 046004 (2017)
  doi:10.1103/PhysRevD.95.046004
  [arXiv:1611.08915 [hep-th]].
  %





\bibitem{Delporte:2019tof} 
  N.~Delporte and V.~Rivasseau,
  ``Perturbative Quantum Field Theory on Random Trees,''
  arXiv:1905.12783 [hep-th].


 
\bibitem{Carrozza:2013wda} 
  S.~Carrozza, D.~Oriti and V.~Rivasseau,
``Renormalization of a SU(2) Tensorial Group Field Theory in
Three Dimensions,''
  Commun.\ Math.\ Phys.\  {\bf 330}, 581 (2014)
  doi:10.1007/s00220-014-1928-x
  [arXiv:1303.6772 [hep-th]].


\bibitem{Geloun:2013saa} 
  J.~Ben Geloun,
``Renormalizable Models in Rank $d\geq 2$ Tensorial Group Field
Theory,''
  Commun.\ Math.\ Phys.\  {\bf 332}, 117 (2014)
  doi:10.1007/s00220-014-2142-6
  [arXiv:1306.1201 [hep-th]].
  
\bibitem{Lahoche:2015tqa} 
  V.~Lahoche and D.~Oriti,
``Renormalization of a tensorial field theory on the homogeneous
space SU(2)/U(1),''
  arXiv:1506.08393 [hep-th].


\bibitem{Lahoche:2015ola} 
  V.~Lahoche, D.~Oriti and V.~Rivasseau,
``Renormalization of an Abelian Tensor Group Field Theory:
Solution at Leading Order,''
  JHEP {\bf 1504}, 095 (2015)
  doi:10.1007/JHEP04(2015)095
  [arXiv:1501.02086 [hep-th]].


\bibitem{Geloun:2012bz} 
  J.~Ben Geloun and E.~R.~Livine,
  ``Some classes of renormalizable tensor models,''
  J.\ Math.\ Phys.\  {\bf 54}, 082303 (2013)
  doi:10.1063/1.4818797
  [arXiv:1207.0416 [hep-th]].


\bibitem{Samary:2012bw} 
  D.~O.~Samary and F.~Vignes-Tourneret,
``Just Renormalizable TGFT's on $U(1)^{d}$ with Gauge
Invariance,''
  Commun.\ Math.\ Phys.\  {\bf 329}, 545 (2014)
  doi:10.1007/s00220-014-1930-3
  [arXiv:1211.2618 [hep-th]].


\bibitem{BenGeloun:2012pu} 
  J.~Ben Geloun and D.~O.~Samary,
``3D Tensor Field Theory: Renormalization and One-loop
$\beta$-functions,''
  Annales Henri Poincare {\bf 14}, 1599 (2013)
  doi:10.1007/s00023-012-0225-5
  [arXiv:1201.0176 [hep-th]].



\bibitem{BenGeloun:2011rc} 
  J.~Ben Geloun and V.~Rivasseau,
  ``A Renormalizable 4-Dimensional Tensor Field Theory,''
  Commun.\ Math.\ Phys.\  {\bf 318}, 69 (2013)
  doi:10.1007/s00220-012-1549-1
  [arXiv:1111.4997 [hep-th]].



\bibitem{Lahoche:2015ola} 
  V.~Lahoche, D.~Oriti and V.~Rivasseau,
``Renormalization of an Abelian Tensor Group Field Theory:
Solution at Leading Order,''
  JHEP {\bf 1504}, 095 (2015)
  doi:10.1007/JHEP04(2015)095
  [arXiv:1501.02086 [hep-th]].

\bibitem{Carrozza:2014rba} 
  S.~Carrozza,
``Discrete Renormalization Group for SU(2) Tensorial Group Field
Theory,''
Ann. Inst. Henri Poincar\'e Comb. Phys. Interact. 2 (2015),
49-112
  doi:10.4171/AIHPD/15
  [arXiv:1407.4615 [hep-th]].


\bibitem{Geloun:2016qyb} 
  J.~B.~Geloun, R.~Martini and D.~Oriti,
``Functional Renormalisation Group analysis of Tensorial Group
Field Theories on $\mathbb{R}^d$,'' Phys. Rev. D \textbf{94} (2016) no.2, 024017
doi:10.1103/PhysRevD.94.024017
[arXiv:1601.08211 [hep-th]].

\bibitem{Geloun:2015qfa} 
  J.~B.~Geloun, R.~Martini and D.~Oriti,
``Functional Renormalization Group analysis of a Tensorial Group
Field Theory on $\mathbb{R}^3$,''
  Europhys.\ Lett.\  {\bf 112}, no. 3, 31001 (2015)
  doi:10.1209/0295-5075/112/31001
  [arXiv:1508.01855 [hep-th]].




\bibitem{Benedetti:2015yaa}
D.~Benedetti and V.~Lahoche,
``Functional Renormalization Group Approach for Tensorial Group Field Theory: A Rank-6 Model with Closure Constraint,''
Class. Quant. Grav. \textbf{33} (2016) no.9, 095003
doi:10.1088/0264-9381/33/9/095003
[arXiv:1508.06384 [hep-th]].

\bibitem{Benedetti:2014qsa} 
  D.~Benedetti, J.~Ben Geloun and D.~Oriti,
``Functional Renormalisation Group Approach for Tensorial Group
Field Theory: a Rank-3 Model,''
  JHEP {\bf 1503}, 084 (2015)
  doi:10.1007/JHEP03(2015)084
  [arXiv:1411.3180 [hep-th]].


\bibitem{Benedetti:2019eyl} 
  D.~Benedetti, R.~Gurau and S.~Harribey,
  ``Line of fixed points in a bosonic tensor model,''
  JHEP {\bf 1906}, 053 (2019)
  doi:10.1007/JHEP06(2019)053
  [arXiv:1903.03578 [hep-th]].


\bibitem{BenGeloun:2018ekd} 
  J.~Ben Geloun, T.~A.~Koslowski, D.~Oriti and A.~D.~Pereira,
``Functional Renormalization Group analysis of rank 3 tensorial
group field theory: The full quartic invariant truncation,''
  Phys.\ Rev.\ D {\bf 97}, no. 12, 126018 (2018)
  doi:10.1103/PhysRevD.97.126018
  [arXiv:1805.01619 [hep-th]].






\bibitem{Carrozza:2016tih} 
  S.~Carrozza and V.~Lahoche,
``Asymptotic safety in three-dimensional SU(2) Group Field
Theory: evidence in the local potential approximation,''
  Class.\ Quant.\ Grav.\  {\bf 34}, no. 11, 115004 (2017)
  doi:10.1088/1361-6382/aa6d90
  [arXiv:1612.02452 [hep-th]].


\bibitem{Lahoche:2016xiq} 
  V.~Lahoche and D.~O. Samary,
``Functional renormalization group for the U(1)-T$_5^6$
tensorial group field theory with closure constraint,''
  Phys.\ Rev.\ D {\bf 95}, no. 4, 045013 (2017)
  doi:10.1103/PhysRevD.95.045013
  [arXiv:1608.00379 [hep-th]].






\bibitem{Carrozza:2017vkz} 
  S.~Carrozza, V.~Lahoche and D.~Oriti,
``Renormalizable Group Field Theory beyond melonic diagrams: an
example in rank four,''
  Phys.\ Rev.\ D {\bf 96}, no. 6, 066007 (2017)
  doi:10.1103/PhysRevD.96.066007
  [arXiv:1703.06729 [gr-qc]].
 
 
\bibitem{Lahoche:2018ggd} 
  V.~Lahoche and D.~O.~Samary,
  ``Ward identity violation for melonic $T^4$-truncation,''
  Nucl.\ Phys.\ B {\bf 940}, 190 (2019)
  doi:10.1016/j.nuclphysb.2019.01.005
  [arXiv:1809.06081 [hep-th]].
  

\bibitem{Lahoche:2018oeo} 
  V.~Lahoche and D.~O.~Samary,
  ``Nonperturbative renormalization group beyond the melonic sector: The effective vertex expansion method for group fields theories,''
  Phys.\ Rev.\ D {\bf 98}, no. 12, 126010 (2018)
  doi:10.1103/PhysRevD.98.126010
  [arXiv:1809.00247 [hep-th]].
 
 

\bibitem{Lahoche:2018vun} 
  V.~Lahoche and D.~O.~Samary,
  ``Unitary symmetry constraints on tensorial group field theory renormalization group flow,''
  Class.\ Quant.\ Grav.\  {\bf 35}, no. 19, 195006 (2018)
  doi:10.1088/1361-6382/aad83f
  [arXiv:1803.09902 [hep-th]].
 
 
\bibitem{Lahoche:2020aeh} 
  V.~Lahoche and D.~O.~Samary,
  ``Pedagogical comments about nonperturbative Ward-constrained melonic renormalization group flow,''
  Phys.\ Rev.\ D {\bf 101}, no. 2, 024001 (2020)
  doi:10.1103/PhysRevD.101.024001
  [arXiv:2001.00934 [hep-th]].

  
\bibitem{Lahoche:2019vzy} 
  V.~Lahoche and D.~O.~Samary,
  ``Ward-constrained melonic renormalization group flow,''
  Phys.\ Lett.\ B {\bf 802}, 135173 (2020)
  doi:10.1016/j.physletb.2019.135173
  [arXiv:1904.05655 [hep-th]].
  
\bibitem{Lahoche:2018hou} 
  V.~Lahoche and D.~O.~Samary,
  ``Progress in the solving nonperturbative renormalization group for tensorial group field theory,''
  Universe {\bf 5}, 86 (2019)
  doi:10.3390/universe5030086
  [arXiv:1812.00905 [hep-th]].
 
  
  
  
\bibitem{Lahoche:2019orv} 
  V.~Lahoche, D.~O.~Samary and A.~D.~Pereira,
  ``Renormalization group flow of coupled tensorial group field theories: Towards the Ising model on random lattices,''
  Phys.\ Rev.\ D {\bf 101}, no. 6, 064014 (2020)
  doi:10.1103/PhysRevD.101.064014
  [arXiv:1911.05173 [hep-th]].
 


\bibitem{Lahoche:2019cxt} 
  V.~Lahoche and D.~O.~Samary,
  ``Ward-constrained melonic renormalization group flow for the rank-four $\phi^6$ tensorial group field theory,''
  Phys.\ Rev.\ D {\bf 100}, no. 8, 086009 (2019)
  doi:10.1103/PhysRevD.100.086009
  [arXiv:1908.03910 [hep-th]].


\bibitem{Lahoche:2019ehm}
V.~Lahoche and D.~O.~Samary,
``Large-$d$ behavior of the Feynman amplitudes for a just-renormalizable tensorial group field theory,''
Phys. Rev. D \textbf{103} (2021) no.8, 085006
doi:10.1103/PhysRevD.103.085006
[arXiv:1911.08601 [hep-th]].



\bibitem{Ooguri:1991ib} 
  H.~Ooguri and N.~Sasakura,
  ``Discrete and continuum approaches to three-dimensional quantum gravity,''
  Mod.\ Phys.\ Lett.\ A {\bf 6}, 3591 (1991)
  doi:10.1142/S0217732391004140
  [hep-th/9108006].

\bibitem{Sasakura:1990fs} 
  N.~Sasakura,
  ``Tensor model for gravity and orientability of manifold,''
  Mod.\ Phys.\ Lett.\ A {\bf 6}, 2613 (1991).
  doi:10.1142/S0217732391003055



\bibitem{Ooguri:1992tw} 
  H.~Ooguri,
  ``Schwinger-Dyson equation in three-dimensional simplicial quantum gravity,''
  Prog.\ Theor.\ Phys.\  {\bf 89}, 1 (1993)
  doi:10.1143/PTP.89.1
  [hep-th/9210028].
  

\bibitem{Boulatov:1992vp} 
  D.~V.~Boulatov,
  ``A Model of three-dimensional lattice gravity,''
  Mod.\ Phys.\ Lett.\ A {\bf 7}, 1629 (1992)
  doi:10.1142/S0217732392001324
  [hep-th/9202074].


\bibitem{Godfrey:1990dt} 
  N.~Godfrey and M.~Gross,
  ``Simplicial quantum gravity in more than two-dimensions,''
  Phys.\ Rev.\ D {\bf 43}, 1749 (1991).
  doi:10.1103/PhysRevD.43.R1749


\bibitem{Gross:1991hx} 
  M.~Gross,
  ``Tensor models and simplicial quantum gravity in  2-D,''
  Nucl.\ Phys.\ Proc.\ Suppl.\  {\bf 25A}, 144 (1992).
  doi:10.1016/S0920-5632(05)80015-5

\bibitem{BenGeloun:2011jnm}
J.~Ben Geloun and V.~Bonzom,
``Radiative corrections in the Boulatov-Ooguri tensor model: The 2-point function,''
Int. J. Theor. Phys. \textbf{50} (2011), 2819-2841
doi:10.1007/s10773-011-0782-2
[arXiv:1101.4294 [hep-th]].

\bibitem{Brezin:1992yc} 
  E.~Brezin and J.~Zinn-Justin,
  ``Renormalization group approach to matrix models,''
  Phys.\ Lett.\ B {\bf 288}, 54 (1992)
  doi:10.1016/0370-2693(92)91953-7
  [hep-th/9206035].
 

\bibitem{DiFrancesco:1993cyw} 
  P.~Di Francesco, P.~H.~Ginsparg and J.~Zinn-Justin,
  ``2-D Gravity and random matrices,''
  Phys.\ Rept.\  {\bf 254}, 1 (1995)
  doi:10.1016/0370-1573(94)00084-G
  [hep-th/9306153].
  
\bibitem{Higuchi:1993pu} 
  S.~Higuchi, C.~Itoi and N.~Sakai,
  ``Renormalization group approach to matrix models and vector models,''
  Prog.\ Theor.\ Phys.\ Suppl.\  {\bf 114}, 53 (1993)
  doi:10.1143/PTPS.114.53
  [hep-th/9307154].



%
%

\bibitem{Ambjorn:1992gw} 
  J.~Ambjorn, L.~Chekhov, C.~F.~Kristjansen and Y.~Makeenko,
  ``Matrix model calculations beyond the spherical limit,''
  Nucl.\ Phys.\ B {\bf 404}, 127 (1993)
  Erratum: [Nucl.\ Phys.\ B {\bf 449}, 681 (1995)]
  doi:10.1016/0550-3213(93)90476-6, 10.1016/0550-3213(95)00391-5
  [hep-th/9302014].
  
 
\bibitem{Ambjorn:1992aw} 
  J.~Ambjorn, J.~Jurkiewicz and C.~F.~Kristjansen,
  ``Quantum gravity, dynamical triangulations and higher derivative regularization,''
  Nucl.\ Phys.\ B {\bf 393}, 601 (1993)
  doi:10.1016/0550-3213(93)90075-Z
  [hep-th/9208032].
  



  
  
 
\bibitem{Higuchi:1994rv} 
  S.~Higuchi, C.~Itoi, S.~Nishigaki and N.~Sakai,
  ``Renormalization group flow in one and two matrix models,''
  Nucl.\ Phys.\ B {\bf 434}, 283 (1995)
  Erratum: [Nucl.\ Phys.\ B {\bf 441}, 405 (1995)]
  doi:10.1016/0550-3213(95)00119-D, 10.1016/0550-3213(94)00437-J
  [hep-th/9409009].
  
\bibitem{Canet:2003qd} 
  L.~Canet, B.~Delamotte, D.~Mouhanna and J.~Vidal,
  ``Nonperturbative renormalization group approach to the Ising model: A Derivative expansion at order partial**4,''
  Phys.\ Rev.\ B {\bf 68}, 064421 (2003)
  doi:10.1103/PhysRevB.68.064421
  [hep-th/0302227].

\bibitem{Alfaro:1992nq} 
  J.~Alfaro and P.~H.~Damgaard,
  ``The D = 1 matrix model and the renormalization group,''
  Phys.\ Lett.\ B {\bf 289}, 342 (1992)
  doi:10.1016/0370-2693(92)91229-3
  [hep-th/9206099].

\bibitem{Stanford:2019vob} 
  D.~Stanford and E.~Witten,
  ``JT Gravity and the Ensembles of Random Matrix Theory,''
  arXiv:1907.03363 [hep-th].


\bibitem{Ambjorn:1991cs} 
  J.~Ambjorn, J.~Jurkiewicz, S.~Varsted, A.~Irback and B.~Petersson,
  ``Critical properties of the dynamical random surface with extrinsic curvature,''
  Phys.\ Lett.\ B {\bf 275}, 295 (1992).
  doi:10.1016/0370-2693(92)91593-X


\bibitem{Ginsparg:1993is} 
  P.~H.~Ginsparg and G.~W.~Moore,
  ``Lectures on 2-D gravity and 2-D string theory,''
  Yale Univ. New Haven - YCTP-P23-92 (92,rec.Apr.93) 197 p. Los Alamos Nat. Lab. - LA-UR-92-3479 (92,rec.Apr.93) 197 p. e: LANL hep-th/9304011
  [hep-th/9304011].

\bibitem{Itoh:2017huk} 
  K.~Itoh,
  ``Gauge symmetry and the functional renormalization group,''
  Int.\ J.\ Mod.\ Phys.\ A {\bf 32}, no. 35, 1747011 (2017).
  doi:10.1142/S0217751X1747011X
  
%
  
\bibitem{Ayala:1993fj} 
  C.~Ayala,
  ``Renormalization group approach to matrix models in two-dimensional quantum gravity,''
  Phys.\ Lett.\ B {\bf 311}, 55 (1993)
  doi:10.1016/0370-2693(93)90533-N
  [hep-th/9304090].
  


  


\bibitem{Sfondrini:2010zm} 
  A.~Sfondrini and T.~A.~Koslowski,
  ``Functional Renormalization of Noncommutative Scalar Field Theory,''
  Int.\ J.\ Mod.\ Phys.\ A {\bf 26}, 4009 (2011)
  doi:10.1142/S0217751X11054048
  [arXiv:1006.5145 [hep-th]].
  
\bibitem{Eichhorn:2014xaa} 
  A.~Eichhorn and T.~Koslowski,
  ``Towards phase transitions between discrete and continuum quantum spacetime from the Renormalization Group,''
  Phys.\ Rev.\ D {\bf 90}, no. 10, 104039 (2014)
  doi:10.1103/PhysRevD.90.104039
  [arXiv:1408.4127 [gr-qc]].
  
  
\bibitem{Eichhorn:2020sla}
A.~Eichhorn, A.~D.~Pereira and A.~G.~A.~Pithis,
``The phase diagram of the multi-matrix model with ABAB-interaction from functional renormalization,''
JHEP \textbf{12} (2020), 131
doi:10.1007/JHEP12(2020)131
[arXiv:2009.05111 [gr-qc]].
  
  
\bibitem{Perez-Sanchez:2020ngw}
C.~I.~Perez-Sanchez,
``Comment on \textquotedblleft{}The phase diagram of the multi-matrix model with ABAB-interaction from functional renormalization\textquotedblright{},''
JHEP \textbf{21} (2020), 042
doi:10.1007/JHEP07(2021)042
[arXiv:2102.06999 [hep-th]].


\bibitem{Perez-Sanchez:2021vpf}
C.~I.~Perez-Sanchez,
``On multimatrix models motivated by random noncommutative geometry II: A Yang-Mills-Higgs matrix model,''
[arXiv:2105.01025 [math-ph]].
  
\bibitem{Eichhorn:2018ylk} 
  A.~Eichhorn, T.~Koslowski, J.~Lumma and A.~D.~Pereira,
  ``Towards background independent quantum gravity with tensor models,''
  doi:10.1088/1361-6382/ab2545
  arXiv:1811.00814 [gr-qc].

\bibitem{Eichhorn:2013isa} 
  A.~Eichhorn and T.~Koslowski,
  ``Continuum limit in matrix models for quantum gravity from the Functional Renormalization Group,''
  Phys.\ Rev.\ D {\bf 88}, 084016 (2013)
  doi:10.1103/PhysRevD.88.084016
  [arXiv:1309.1690 [gr-qc]].
  
  
\bibitem{Lahoche:2019ocf} 
  V.~Lahoche and D.~O.~Samary,
  ``Revisited functional renormalization group approach for random matrices in the large-$N$ limit,''
  Phys.\ Rev.\ D {\bf 101}, 106015 (2020)
  doi:10.1103/PhysRevD.101.106015
  [arXiv:1909.03327 [hep-th]].
 
 
\bibitem{Lahoche:2020pjo}
V.~Lahoche and D.~O.~Samary,
``Reliability of the local truncations for the random tensor models renormalization group flow,''
Phys. Rev. D \textbf{102} (2020) no.5, 056002
doi:10.1103/PhysRevD.102.056002
[arXiv:2005.11846 [hep-th]].


  
 
  

\bibitem{Gurau:2015tua} 
  R.~Gurau, A.~Tanasa and D.~R.~Youmans,
  ``The double scaling limit of the multi-orientable tensor model,''
  EPL {\bf 111}, no. 2, 21002 (2015)
  doi:10.1209/0295-5075/111/21002
  [arXiv:1505.00586 [hep-th]].



\bibitem{Eichhorn:2017xhy} 
  A.~Eichhorn and T.~Koslowski,
  ``Flowing to the continuum in discrete tensor models for quantum gravity,''
  Ann.\ Inst.\ H.\ Poincare Comb.\ Phys.\ Interact.\  {\bf 5}, no. 2, 173 (2018)
  doi:10.4171/AIHPD/52
  [arXiv:1701.03029 [gr-qc]].
  %


\bibitem{Dijkgraaf:1989pz}
R.~Dijkgraaf and E.~Witten,
``Topological Gauge Theories and Group Cohomology,''
Commun. Math. Phys. \textbf{129} (1990), 393
doi:10.1007/BF02096988


\bibitem{Putrov:2016qdo}
P.~Putrov, J.~Wang and S.~T.~Yau,
``Braiding Statistics and Link Invariants of Bosonic/Fermionic Topological Quantum Matter in 2+1 and 3+1 dimensions,''
Annals Phys. \textbf{384} (2017), 254-287
doi:10.1016/j.aop.2017.06.019
[arXiv:1612.09298 [cond-mat.str-el]].
%



\bibitem{Samary:2014oya} 
  D.~O.~ Samary, C.~I.~P\'erez-S\'anchez, F.~Vignes-Tourneret and R.~Wulkenhaar,
  ``Correlation functions of a just renormalizable tensorial group field theory: the melonic approximation,''
  Class.\ Quant.\ Grav.\  {\bf 32}, no. 17, 175012 (2015)
  doi:10.1088/0264-9381/32/17/175012
  [arXiv:1411.7213 [hep-th]].





\bibitem{Samary:2014tja} 
  D.~O.~Samary,
  ``Closed equations of the two-point functions for tensorial group field theory,''
  Class.\ Quant.\ Grav.\  {\bf 31}, 185005 (2014)
  doi:10.1088/0264-9381/31/18/185005
  [arXiv:1401.2096 [hep-th]].



\bibitem{Perez-Sanchez:2016zbh} 
  C.~I.~P\'erez-S\'anchez,
  ``The full Ward-Takahashi Identity for colored tensor models,''
  Commun.\ Math.\ Phys.\  {\bf 358}, no. 2, 589 (2018)
  doi:10.1007/s00220-018-3103-2
  [arXiv:1608.08134 [math-ph]].



\bibitem{Morris:2016spn}
T.~R.~Morris,
``Large curvature and background scale independence in single-metric approximations to asymptotic safety,''
JHEP \textbf{11} (2016), 160
doi:10.1007/JHEP11(2016)160
[arXiv:1610.03081 [hep-th]].

\bibitem{Ellwanger:1994iz}
U.~Ellwanger,
``Flow equations and BRS invariance for Yang-Mills theories,''
Phys. Lett. B \textbf{335} (1994), 364-370
doi:10.1016/0370-2693(94)90365-4
[arXiv:hep-th/9402077 [hep-th]].
%


\bibitem{Wetterich:1991be} 
  C.~Wetterich,
 ``The Average action for scalar fields near phase transitions,''
  Z.\ Phys.\ C {\bf 57}, 451 (1993).
  doi:10.1007/BF01474340


\bibitem{Wetterich:1992yh} 
  C.~Wetterich,
  ``Exact evolution equation for the effective potential,''
  Phys.\ Lett.\ B {\bf 301}, 90 (1993)
  doi:10.1016/0370-2693(93)90726-X
  [arXiv:1710.05815 [hep-th]].
  %


  
  
\bibitem{Wetterich:2016ewc} 
  C.~Wetterich,
  ``Gauge invariant flow equation,''
  Nucl.\ Phys.\ B {\bf 931}, 262 (2018)
  doi:10.1016/j.nuclphysb.2018.04.020
  [arXiv:1607.02989 [hep-th]].


\bibitem{Freire:2000mn} 
  F.~Freire, D.~F.~Litim and J.~M.~Pawlowski,
  ``Gauge invariance, background fields and modified ward identities,''
  Int.\ J.\ Mod.\ Phys.\ A {\bf 16}, 2035 (2001)
  doi:10.1142/S0217751X01004669
  [hep-th/0101108].

\bibitem{Wetterich:2017aoy} 
  C.~Wetterich,
  ``Gauge-invariant fields and flow equations for Yang–Mills theories,''
  Nucl.\ Phys.\ B {\bf 934}, 265 (2018)
  doi:10.1016/j.nuclphysb.2018.07.002
  [arXiv:1710.02494 [hep-th]].


\bibitem{Morris:1993qb} 
  T.~R.~Morris,
  ``The Exact renormalization group and approximate solutions,''
  Int.\ J.\ Mod.\ Phys.\ A {\bf 9}, 2411 (1994)
  doi:10.1142/S0217751X94000972
  [hep-ph/9308265].

\bibitem{Morris:2005ck} 
  T.~R.~Morris,
  ``Equivalence of local potential approximations,''
  JHEP {\bf 0507}, 027 (2005)
  doi:10.1088/1126-6708/2005/07/027
  [hep-th/0503161].


\bibitem{Morris:2000hm} 
  T.~R.~Morris and J.~F.~Tighe,
  ``Convergence of derivative expansions in scalar field theory,''
  Int.\ J.\ Mod.\ Phys.\ A {\bf 16}, 2095 (2001)
  doi:10.1142/S0217751X01004761
  [hep-th/0102027].








\bibitem{Litim:2000ci} 
  D.~F.~Litim,
  ``Optimization of the exact renormalization group,''
  Phys.\ Lett.\ B {\bf 486}, 92 (2000)
 doi:10.1016/S0370-2693(00)00748-6
  [hep-th/0005245].
  


\bibitem{Litim:2001dt} 
  D.~F.~Litim,
  ``Derivative expansion and renormalization group flows,''
  JHEP {\bf 0111}, 059 (2001)
  doi:10.1088/1126-6708/2001/11/059
  [hep-th/0111159].
  
  
\bibitem{Canet:2002gs} 
  L.~Canet, B.~Delamotte, D.~Mouhanna and J.~Vidal,
  ``Optimization of the derivative expansion in the nonperturbative renormalization group,''
  Phys.\ Rev.\ D {\bf 67}, 065004 (2003)
  doi:10.1103/PhysRevD.67.065004
  [hep-th/0211055].
  
  
 
\bibitem{Delamotte:2007pf} 
  B.~Delamotte,
  ``An Introduction to the nonperturbative renormalization group,''
  Lect.\ Notes Phys.\  {\bf 852}, 49 (2012)
  doi:10.1007/978-3-642-27320-$9_2$
  [cond-mat/0702365 [cond-mat.stat-mech]].






\bibitem{Lahoche:2015yya}
V.~Lahoche,
``Constructive Tensorial Group Field Theory I: The $U(1)-T^4_3$ Model,''
J. Phys. A \textbf{51} (2018) no.18, 185403
doi:10.1088/1751-8121/aab8a8
[arXiv:1510.05050 [hep-th]].


\bibitem{Lahoche:2015zya}
V.~Lahoche,
``Constructive Tensorial Group Field Theory II: The $U(1)-T^4_4$ Model,''
J. Phys. A \textbf{51} (2018) no.18, 185402
doi:10.1088/1751-8121/aab8a7
[arXiv:1510.05051 [hep-th]].





 
\end{thebibliography}
\end{document}